\documentclass{ieeeaccess}
\usepackage[utf8]{inputenc}
\usepackage{bm}
\usepackage{amsmath}
\usepackage{amsfonts}
\usepackage{amssymb}
\usepackage{algorithmic}
\usepackage{amsthm}
\usepackage{hyperref}
\usepackage{physics}
\usepackage{mathrsfs}
\usepackage{cite}
\usepackage{graphicx}
\usepackage{ulem}
\usepackage{caption}
\usepackage{subcaption}

\usepackage{nicefrac}
\usepackage{placeins}
\usepackage{textcomp}
\newtheorem{theorem}{Theorem}

\newtheorem{lemma}{Lemma}
\newcommand{\etal}{\textit{et al.}}
\allowdisplaybreaks[4]
\usepackage[T1]{fontenc}
\newcommand{\sumtwo}{\operatorname*{\sum~\sum}}

\usepackage{regexpatch}

\def\BibTeX{{\rm B\kern-.05em{\sc i\kern-.025em b}\kern-.08em
    T\kern-.1667em\lower.7ex\hbox{E}\kern-.125emX}}
\makeatother
\begin{document}
\history{Date of publication xxxx 00, 0000, date of current version xxxx 00, 0000.}
\doi{10.1109/TQE.2020.DOI}
\title{Expressivity of Variational Quantum Machine Learning on the Boolean Cube}
\author{\uppercase{Dylan Herman}\authorrefmark{1},
\uppercase{Rudy Raymond}\authorrefmark{2,}\authorrefmark{3},
\uppercase{Muyuan Li}\authorrefmark{4},
\uppercase{Nicolas Robles}\authorrefmark{4},
\\\uppercase{Antonio Mezzacapo}\authorrefmark{4},
\uppercase{Marco Pistoia}\authorrefmark{1}}
\address[1]{Global Technology Applied Research, JPMorgan Chase, New York, New York 10017, USA}
\address[2]{IBM Quantum, IBM Research - Tokyo, 19-21 Nihonbashi Hakozaki-cho, Chuo-ku, Tokyo, 103-8510, Japan}
\address[3]{Quantum Computing Center, Keio University, 3-14-1 Hiyoshi, Kohoku-ku, Yokohama, Kanagawa, 223-8522, Japan}
\address[4]{IBM Quantum, IBM T.~J.~Watson Research Center, Yorktown Heights, New York 10598, USA}

\corresp{Corresponding author: Dylan Herman (email: dylan.a.herman@jpmorgan.com).}

\begin{abstract}
Categorical data plays an important part in machine learning research and appears in a variety of applications. Models that can express large classes of real-valued functions on the Boolean cube are useful for problems involving discrete-valued data types, including those which are not Boolean. To this date, the commonly used schemes for embedding classical data into variational quantum machine learning models encode continuous values. 
Here we investigate quantum embeddings for encoding Boolean-valued data into parameterized quantum circuits used for machine learning tasks. 
We narrow down representability conditions for functions on the $n$-dimensional Boolean cube with respect to previously known results, using two quantum embeddings: a phase embedding and an embedding based on quantum random access codes. We show that for any real-valued function on the $n$-dimensional Boolean cube, there exists a variational linear quantum model based on a phase embedding using $n$ qubits that can represent it and an ensemble of such models using $d < n$ qubits that can express any function with degree at most $d$. Additionally, we prove that variational linear quantum models that use the quantum random access code embedding can express functions on the Boolean cube with degree $ d\leq \lceil\frac{n}{3}\rceil$ using $\lceil\frac{n}{3}\rceil$ qubits, and that an ensemble of such models can represent any function on the Boolean cube with degree $ d\leq \lceil\frac{n}{3}\rceil$. Furthermore, we discuss the potential benefits of each embedding and the impact of serial repetitions. Lastly, we demonstrate the use of the embeddings presented by performing numerical simulations and experiments on IBM quantum processors using the Qiskit machine learning framework.
\end{abstract}

\begin{keywords}
quantum machine learning, variational quantum algorithms, expressivity, Fourier analysis, Boolean cube
\end{keywords}

\titlepgskip=-15pt

\maketitle

\section{Introduction}
\label{sec:intro}
Machine learning problems involving categorical data are prevalent across many domains. The range of a categorical variable lies in a finite set, and each element in this set can be associated with an integer. Thus one can map a single categorical variable to multiple binary variables. If our goal is to perform supervised learning, then this converts the problem into learning a real-valued function on the $n$\textit{-dimensional Boolean (hyper)cube} $\mathbb{B}^n = \{0, 1\}^n$. This implies that models that can express large classes of functions on the Boolean cube are useful for problems involving categorical data.In this article, we consider using \emph{variational quantum machine learning} (VQML) \cite{schuld2018supervisedBook} to fit real-valued functions of multiple binary variables. Thus, such models can be applied to regression or classification tasks. Beyond machine learning, variational quantum algorithms \cite{2021vqa} have been applied to chemistry  \cite{peruzzo2014variational, kandala2017hardware, rattew2019domain}, combinatorial optimization \cite{farhi2014quantum, Hadfield_2019}, quantum linear systems \cite{vqls, Huang_2021}, and the simulation of quantum dynamics \cite{McArdle_2019, Yuan_2019}. When applied to supervised learning, VQML consists of using parameterized quantum circuits (PQCs) built from two types of circuit blocks: \emph{embedding blocks}, which encode the inputs into a quantum system, and \textit{trainable blocks}, where learnable parameters are adjusted in order to optimize the output results. There are various methods for optimizing the learnable parameters in a hybrid classical-quantum iterative manner, including analytical gradients \cite{schuld2019evaluating,stokes2020quantum,koczor2020quantum, Ostaszewski_2021, watanabe2021optimizing}. This paradigm has been used to construct various analogues to classical machine learning models applicable to supervised learning tasks \cite{farhi2018classification, Mitarai_2018, Havl_ek_2019,2019schuldKilloran, benedetti2019parameterized, Cong_2019, 2021cnn}.

As an example of the potential applicability of VQML, it has been observed that these models can be used to solve a variety of financial problems, such as fraud detection and creditworthiness determination \cite{pistoia2021quantum, herman2022survey, Egger_2020}. There has been an active line of studies to characterize the expressivity \cite{larocca2021theory, 2021tobias, Du_2022}, the generalizability \cite{2021powerofqnn, caro2022pseudo, Bu_2022, Chen2021,  popescu2021, Caro2021encodingdependent, Banchi_2021,   caro2022generalization, cai2022sample}, and the trainability \cite{mcclean2018barren, zhang2020toward, 2021Patti} of VQML models. However, the specific case of quantum models with discrete-valued inputs has not been investigated as extensively \cite{thumwanit2021trainable, yano2020efficient}.

A recent study by Schuld \etal~\cite{Schuld_2021} showed that the output, when it is represented by the expected value of an observable, of a VQML model can be expressed as a partial sum of a multidimensional Fourier series. The connection between VQML models and Fourier series was also observed in \cite{gil2020input}. Recently, Caro \etal~\cite{Caro2021encodingdependent} derived generalization bounds for such models. The range of attainable frequencies is related to the quantum embedding used, and this range can be broadened by repeating the embedding sequentially, a process called \emph{data re-uploading} \cite{2020reupload}, or by introducing additional sets of qubits and repeating the embedding in parallel for each set. The observable and trainable blocks control the coefficients of the Fourier basis elements in the partial sum. Since every function in $L_2([0, 2\pi]^n)$ can be represented by the limit of a Fourier series \cite{weisz2012summability}, VQML models can approximate any function in this space to arbitrarily small error, in $L_2$ norm, by using an embedding scheme that produces the required Fourier spectrum. This also assumes that the observable and trainable blocks can fit the Fourier coefficients to the desired error, which Schuld \etal~assumed when deriving their results. Along similar lines, Goto \etal~\cite{2021uapfeaturemaps} demonstrated that models built from a linear combination of basis functions derived from quantum-enhanced feature spaces are universal for continuous functions. Similar to Schuld \etal, the embeddings that Goto \etal~used consisted of serial and parallel repetitions of simple encoding schemes. We show that \emph{variational linear quantum models}, which do not make use of serial or parallel repetitions of a quantum embedding, are sufficient for representing functions on the Boolean cube. Variational linear quantum models use PQCs that consist of one embedding block and one trainable block. For two quantum embeddings, we use Fourier analysis to derive the classes of real-valued functions on the Boolean cube that can be represented by variational linear quantum models. The number of qubits used only depends on the dimension of the input. 

In this paper, we explore two research directions.
\begin{enumerate}
    \item First, we consider a \emph{phase embedding}, which encodes each input bit into the relative phase of a single-qubit state, i.e. the number of qubits used equals the number of input bits. We show that any real-valued function on the Boolean cube, $\mathbb{B}^{n}$, can be represented by a variational linear quantum model that uses the phase embedding and that a classical ensemble, formed by summing the outputs of multiple models each using $d$ qubits, can express any function with degree $\leq d$. The degree of a function on $\mathbb{B}^n$ is the maximal Hamming weight over all $\bm{s} \in \mathbb{B}^n$ where the Fourier transform is nonzero.
    
    \item Then, we further consider a \emph{QRAC embedding}---a quantum embedding that makes use of \emph{quantum random access codes} (QRACs) \cite{ambainisqrac, Doriguello_2021}. This embedding was introduced by Yano \etal~\cite{yano2020efficient} for encoding categorical data into variational quantum classifiers. We investigate the classes of functions expressible by variational linear quantum models using this embedding. Specifically, we show that any function of degree $d \leq \lceil \frac{n}{3} \rceil$ can be represented by a classical ensemble formed by summing the outputs of multiple QRAC-embedding-based variational linear quantum models each using $\lceil \frac{n}{3} \rceil$ qubits. 
\end{enumerate}
We note that the above results imply that for functions with degree $\leq \lceil \frac{n}{3} \rceil$ an ensemble of phase embedding models requires only  $\lceil \frac{n}{3} \rceil$ qubits, which is the same number of qubits used with the QRAC embedding. However, it can still be beneficial to use the QRAC embedding in certain cases as discussed later. 

Juntas form an important class of functions on Boolean domains. A $k$\textit{-junta} is a function, Boolean or real-valued, that depends on at most $k$ out of the $n$ input bits. These functions are useful in computational learning theory \cite{kearns1994introduction} for modeling learning tasks where the data can be explained using a subset of the available features \cite{blum1994relevant}. Such scenarios typically occur when applying supervised learning to real-world data sets \cite{mossel2003learning}. There has been a lot of progress in developing quantum computational learning theory \cite{arunachalam2017guest}. For example, there exist algorithms in the query model \cite{ambqunatumquery} for both learning and testing $k$-juntas \cite{At_c__2007, belovs, groupcomb, Arunachalam2021twonewresultsabout}, some of which make use of both quantum and classical queries. By definition, if $k \leq \lceil \frac{n}{3}\rceil$, then the degree of the junta is guaranteed to be at most $\lceil \frac{n}{3}\rceil$ so the junta can be represented by  an ensemble of  linear quantum models that use the QRAC or phase embedding.

With regards to classical neural networks, there has been recent work investigating the learnability of parity functions \cite{daniely2020learning} and real-valued functions on the Boolean cube \cite{Yang2019}. There are also neural networks for Lattice Regression~\cite{Gupta2016} and the recent Hierarchical Lattice Layer~\cite{Yanagisawa2022} for partially monotone regression. %

We perform  experiments on simulators and on IBM Quantum hardware to study the expressiveness of variational linear quantum models for low-degree functions. These experiments demonstrate the efficacy of the phase and QRAC embeddings for representing functions on the Boolean cube.

\subsection{Main Results}
Summarizing, we list here the main contributions of this paper.

\begin{enumerate}

    \item We show that for any function on the Boolean cube $\mathbb{B}^n$, there exists a variational linear quantum model with $n$ qubits based on a phase embedding such that the output of the quantum model agrees with the target function for all inputs. Additionally, we show that for any function with degree $\leq d$ there exists an ensemble of variational linear quantum models using the phase embedding and $d$ qubits such that the output of the ensemble agrees with the output of the target function for all inputs.
    \item We then present sufficient conditions for variational linear quantum models using a QRAC embedding to be able to express functions of degree $d \leq \lceil \frac{n}{3} \rceil$ using $\lceil\frac{n}{3} \rceil$ qubits. Moreover, we then demonstrate that for any function of degree $d \leq \lceil \frac{n}{3} \rceil$ on the Boolean cube $\mathbb{B}^n$, there exists an ensemble of QRAC-based variational linear quantum models with $\lceil\frac{n}{3} \rceil$ qubits each such that the output of the ensemble agrees with the output of the target function for all inputs.
    \item We test these two embeddings on low-degree functions on the Boolean cube via numerical experiments and on IBM superconducting quantum processors.

\end{enumerate}
We note that results derived for the phase and QRAC embeddings were proven under the assumption of universal trainable gates and arbitrary observables that are diagonal
in the computational basis.

\subsection{Paper Organization}
Section \ref{sec:prelims} reviews the Fourier analysis of functions with Boolean inputs and the use of PQCs for machine learning. Section \ref{sec:embeddings} introduces embeddings for representing functions on the Boolean cube with PQCs. Then, we use tools from Fourier analysis to study the expressivity of variational quantum machine learning models that make use of these embeddings. In Section \ref{sec:experiments} we apply variational quantum models, using either the phase or QRAC embeddings, to supervised learning problems involving low-degree functions on the Boolean cube. These experiments were run in simulation and on IBM Quantum hardware. Lastly, the appendices contain further computational elaborations of the topics discussed in the main text.

\section{Preliminaries}
This section introduces the concepts necessary to understand the novel contributions of this paper.  Particularly, it focuses on the Fourier analysis on the Boolean cube and gives an overview of the state of the art of variational quantum machine learning.
\label{sec:prelims}
\subsection{Fourier Analysis on the Boolean Cube}
\label{sec:fourier_onB}
First, we briefly review the Fourier analysis of real-valued functions with Boolean inputs. This short review is based on the introduction by de Wolf~\cite{Wolf2008ABI}. We consider the $2^n$-dimensional real vector space $\mathcal{G} := \{f: \mathbb{B}^n \rightarrow \mathbb{R}\}$. This space can be equipped with the following inner product:
\begin{equation}
\label{eqn:innerproduct}
    \left<f, g\right> := \frac{1}{2^n} \sum_{\bm{b}\in \mathbb{B}^n} f(\bm{b})g(\bm{b}), ~~\forall f,g \in \mathcal{G}.
\end{equation}
To every tuple $\bm{s} \in \mathbb{B}^n$, we associate a function $\chi_{\bm{s}} : \mathbb{B}^n \xrightarrow[]{} \{\pm 1\}$ that is defined as follows:
\begin{equation}
\label{eqn:chi}
    \chi_{\bm{s}}(\bm{b}) := \prod_{i:s_i=1}(-1)^{b_{i}} = (-1)^{\bm{s}\cdot \bm{b}},
\end{equation}
where $\bm{s}\cdot \bm{b}$ is given by the scalar product:
\begin{align}
    \bm{s}\cdot\bm{b} := \sum_{i=1}^{n}s_ib_i.
\end{align}
The function $\chi_{\bm{s}}$ depends on the parity of a subset, indicated by $\bm{s}$, of the input bits.

With respect to the inner product defined in \eqref{eqn:innerproduct}, the set containing all $\chi_{\bm{s}}$ forms an orthonormal basis for  $\mathcal{G}$, called the \emph{Fourier basis}. The \emph{Fourier transform} of a given $f \in \mathcal{G}$, denoted by $\widehat{f}$, is defined as follows:
\begin{equation}
\label{eqn:fouriercoeff}
\widehat{f}(\bm{s}) := \left<f, \chi_{\bm{s}}\right> = \frac{1}{2^n} \sum_{\bm{b}\in \mathbb{B}^n} f(\bm{b})\chi_{\bm{s}}(\bm{b}),~~\forall \bm{s} \in \mathbb{B}^n.
\end{equation}
Because the set of all $\chi_{\bm{s}}$ forms an orthonormal basis, it holds that any $f$ can be expressed as
\begin{equation}
\label{eqn:fourierseries}
f(\bm{b}) = \sum_{\bm{s} \in \mathbb{B}^n} \widehat{f}(\bm{s}) \chi_{\bm{s}}(\bm{b}),
\end{equation}
where the value $\widehat{f}(\bm{s})$ is the \emph{Fourier coefficient} associated with $\chi_{\bm{s}}$ and the set of all $\widehat{f}(\bm{s})$ is called the \emph{Fourier spectrum} of $f$. The \emph{degree} of $f$ is the maximal Hamming weight over all $\bm{s} \in \mathbb{B}^n$ such that $\widehat{f}(\bm{s}) \neq 0$. 

In Section \ref{sec:intro}, we introduced a $k$\textit{-junta} as a function on $\mathbb{B}^n$ whose output only depends on $k$ of the $n$ input variables $b_1, b_2, \dots, b_n$. For a given $k$-junta, suppose $\mathcal{C} \subseteq [n] := \{1, \dots, n \}$ contains the $k$ indices corresponding to the input variables that the junta depends on. It can be easily shown that the Fourier transform of a $k$-junta can only be nonzero on elements from the set 
\begin{align}
\{ \bm{s} \in \mathbb{B}^n ~|~ \forall i \in [n]: s_i =1 \implies i \in \mathcal{C}\}.
\end{align}
Thus, the degree of the junta is bounded by $\lvert\mathcal{C}\rvert = k$, so when $k \ll n$, a $k$-junta is guaranteed to also be a low-degree function. In the following sections, we will use these definitions to analyze the classes of functions on $\mathbb{B}^n$ that can be expressed by VQML models.

\subsection{Variational Quantum Machine Learning}
\label{sec:quantumcircuitlearning}
This section reviews relevant concepts of VQML~\cite{schuld2018supervisedBook}. Before moving to functions defined on the Boolean cube, we consider the  task of fitting a real-valued function that is defined on an arbitrary set $\mathcal{A} \subseteq \mathbb{R}^n$. For a continuous-valued range, this task is called \textit{regression}, and for a discrete-valued range, it is called \textit{classification}. The input data $\bm{x} \in \mathcal{A}$, stored on a classical memory, can be  embedded into a quantum state by utilizing an $m$-qubit parameterized unitary operator $\bm{U}(\bm{x})$, which is a unitary-operator-valued function  of the $n$-dimensional vectors in $\mathcal{A}$. The operator $\bm{U}$ is called a \emph{quantum embedding}.

Let us consider a parameter-independent Hermitian observable $\bm{D}$, defined to be diagonal with respect to the computational basis, and thus:
\begin{equation}
    \bm{D} \in \textup{span}_\mathbb{R}\{ \mathbb{I}, \mathsf{Z}\}^{\otimes m}
\end{equation}
We further define a parameterized observable $\bm{O}_{\bm{\theta}}$ with variational parameters $\bm{\theta}$ as follows:
\begin{equation}
\label{eqn:param_obs}
    \bm{O}_{\bm{\theta}} := \bm{W}^{\dagger}(\bm{\theta})\bm{D}\bm{W}(\bm{\theta}),
\end{equation}
where $\bm{W}(\bm{\theta})$ is a unitary operator implemented by a PQC.

The VQML model that we focus on in this work is the \emph{variational linear quantum model}:
\begin{equation}
    \label{eqn:variational_linear_quantum_model}
    f_{\bm{\theta}}(\bm{x}) := \Tr[\bm{O}_{\bm{\theta}}\rho(\bm{x})],
\end{equation}
where $\rho(\bm{x}) := \bm{U}(\bm{x})\ketbra{0_m}\bm{U}^{\dagger}(\bm{x})$ is the state of the system after the action of the unitary $\bm{U}(\bm{x})$ and $\Tr$ is the trace operator.  Essentially, $f_{\bm{\theta}}(\bm{x})$
maps $\bm{x}$ to a real number by taking the expectation of $\bm{O}_{\bm{\theta}}$ with respect to $\rho(\bm{x})$.
 This model is also called a \textit{quantum neural network} \cite{2021powerofqnn}, and in the context of classification, it has been called the \textit{explicit linear quantum classifier} \cite{gyurik2021structural} or \textit{variational quantum classifier} \cite{Havl_ek_2019}. Since the expectation of an observable is continuous valued, for classification, some post-processing of the output is required to map it to the finite set of possible classes.

We can implement the parameterized measurement by evolving $\rho(\bm{x})$ by the parameterized unitary operator $\bm{W}(\bm{\theta})$ and then measuring $\bm{D}$. We relate this sequence of operations to Equation \eqref{eqn:variational_linear_quantum_model} as follows:
\begin{align}
    &\Tr[\bm{D}\bm{W}(\bm{\theta})\rho(\bm{x})\bm{W}^{\dagger}(\bm{\theta})] \nonumber\\&= \Tr[\bm{W}^{\dagger}(\bm{\theta})\bm{D}\bm{W}(\bm{\theta})\rho(\bm{x})] = \Tr[\bm{O}_{\bm{\theta}}\rho(\bm{x})],
\end{align}
where we used the cyclic property of the trace.

The model in Equation \eqref{eqn:variational_linear_quantum_model} is linear in the sense of being a quantum analogue to the linear models \cite{friedman2001elements} of classical machine learning \cite{Havl_ek_2019, schuld2021supervised}. Explicitly, a linear model in classical machine learning is of the form:
\begin{equation}
\label{eqn:classical_linear_model}
    h(\bm{x}) = \left<\bm{w}, \phi(\bm{x})\right>_{\mathcal{F}},
\end{equation}
where $\phi: \mathcal{A} \xrightarrow[]{} \mathcal{F}$ is called a \emph{feature map}, and $\mathcal{F}$ is the associated \emph{feature space}. In addition, $\bm{w} \in \mathcal{F}$ is fixed for all inputs $\bm{x} \in \mathcal{A}$, but it is chosen to minimize some cost function by using an optimization procedure. For the model defined in Equation \eqref{eqn:variational_linear_quantum_model}, $\phi(\bm{x}) = \rho(\bm{x})$ maps $\bm{x}$ into a \emph{quantum feature space}, which contains $2^m \times 2^m$ density matrices representing quantum states called \emph{feature states}. In addition, the observable $\bm{O}_{\bm{\theta}}$ represents a $\bm{\theta}$-parameterized family of $\bm{w}$'s. In this case, each $\bm{w}$ in the family is a Hermitian matrix. The inner product in Equation \eqref{eqn:classical_linear_model} is now the Hilbert-Schmidt inner product. Lastly, a VQML model that interleaves embedding and trainable layers, which by \eqref{eqn:variational_linear_quantum_model} implies it is not a linear model, is known in literature as a \emph{data re-uploading model} \cite{2020reupload}. %

The expressivity of both classical and quantum linear models solely depends on the feature map used to encode $\bm{x}$, as both $\bm{w}$ and $\bm{O}_{\bm{\theta}}$ only define linear functions in the feature space. The feature maps can be used to make $h$ or $f_{\bm{\theta}}$ nonlinear as functions on the domain $\mathcal{A}$. In supervised learning,
the goal is to minimize the \emph{regularized empirical risk}:
\begin{equation}
\label{eqn:regularized_risk}
      \frac{1}{\lvert\mathcal{T}\rvert}\sum_{(\bm{x}, y) \in \mathcal{T}} \ell(\bm{x}, y, h) + J(h)
\end{equation}
over the labeled training set
\begin{equation}
\label{eqn:training_dataset}
    \mathcal{T} := \{ (\bm{x}_1, y_1), \dots, (\bm{x}_t, y_t)\},
\end{equation}
where each $\bm{x}_i \in \mathcal{A}$ and each $y_i \in \mathbb{R}$. 
The functions $\ell$ and $J$ are called the \emph{loss} and \emph{regularizer}, respectively.

Classically, when $\phi(\bm{x})$ is difficult to compute explicitly or to operate on, we instead utilize \emph{kernel methods}. The \emph{kernel function} induced by the feature map $\phi$ is defined as:\begin{equation}
k(\bm{x}_i, \bm{x}_j) := \left<\phi(\bm{x}_i), \phi(\bm{x}_j)\right>_{\mathcal{F}}.
\label{eqn:featuremap}
\end{equation}
Kernel methods consider $h$, in Equation \eqref{eqn:classical_linear_model}, as a function in the \emph{reproducing-kernel Hilbert space} (RKHS) generated by $k$.
This \textit{kernel trick} is effective when the functional form of $k$ is easier to evaluate than it is to explicitly compute the inner product between  $\phi(\bm{x}_i)$ and $\phi(\bm{x}_j)$ as done in Equation \eqref{eqn:featuremap}. A common classical example is the Gaussian kernel, which efficiently computes inner products in an infinite-dimensional feature space \cite{friedman2001elements}. Suppose $J(h) := z(\lVert h\rVert_{k})$, where  $z : [0, \infty) \xrightarrow[]{} \mathbb{R}$ is strictly increasing, and $\lVert h\rVert_{k}$ is the norm of $h$ in the RKHS. Then, according to the representer theorem \cite{hofmann2008kernel}, any minimizer $h_{\min}$ of the regularized empirical risk \eqref{eqn:regularized_risk} lies in the RKHS and is of the form:
\begin{align}
\label{eqn:representer}
    &h_{\min}(\bm{x}') = \sum_{(\bm{x}, y) \in \mathcal{T}}\alpha_{\bm{x}} k(\bm{x}, \bm{x}') \nonumber\\&= \sum_{(\bm{x}, y) \in \mathcal{T}}\alpha_{\bm{x}}\left<\phi(\bm{x}), \phi(\bm{x}')\right>_{\mathcal{F}}.
\end{align}
If $\ell$ is convex and $J(h) = \beta \lVert h\rVert_k^2$, where $\beta \ge 0$, then the $\alpha_{\bm{x}}$'s can be found by solving a convex optimization problem called \emph{kernel ridge regression}. This requires computing the \emph{kernel matrix} $\bm{K}$, with entries $\bm{K}_{i,j} = k(\bm{x}_i, \bm{x}_j)$. In addition, $\bm{K}$ is a $\lvert\mathcal{T}\rvert \times \lvert\mathcal{T}\rvert$ real-valued symmetric matrix. 

The \emph{quantum-kernel method}, as originally stated \cite{Havl_ek_2019}, uses the \emph{fidelity kernel}: \begin{align}
\label{eqn:quantum_kernel}
    k_{Q}(\bm{x}_i, \bm{x}_j) := \Tr[\rho(\bm{x}_i) \rho(\bm{x}_j)] = \lvert\braket{\psi_i}{\psi_j}\rvert^2,
\end{align} 
where $\ket{\psi_i} = \bm{U}(\bm{x}_i)\ket{0_m}$, and $\bm{U}$ is a quantum embedding. This kernel computes the Hilbert-Schmidt inner product between quantum feature states. Liu \etal~\cite{2021liurigorous} demonstrated a quantum speedup using such kernel methods for solving a discrete-log-inspired supervised learning problem. It has been observed that generalization can be difficult with the fidelity kernel, however, there exist heuristics \cite{2021Huang, kubler2021} and hyperparameter optimization techniques \cite{Canatar2022} to  enable generalization.  The kernel defined in Equation \eqref{eqn:quantum_kernel} is evaluated on a quantum device for all pairs of training data elements, which avoids performing classical operations on the $2^m$-dimensional statevectors $\ket{\psi_i}$. The entries of the corresponding kernel matrix are given to a classical computer to find the $\alpha_{\bm{x}}$'s, which involves solving a convex optimization problem. Alternatively, the training procedure for VQML models does not consist of computing the kernel. Instead of finding the $\alpha_{\bm{x}}$'s, we optimize the variational parameters $\bm{\theta}$, which can be a non-convex problem. This problem can still be non-convex regardless of whether the loss and regularization functions are convex \cite{2021nonconvexloss, Rivera_Dean_2021}. Furthermore, the quantum-kernel method has access to all minimizers, $h_{\text{min}}$, of the regularized empirical risk \eqref{eqn:regularized_risk}, which lie in the RKHS generated by $k_{Q}$. In contrast, the choices one makes for the PQC $\bm{W}(\bm{\theta})$ and the observable $\bm{D}$ restrict the set of functions that Equation \eqref{eqn:variational_linear_quantum_model} can represent to a subset of the RKHS, which may not contain $h_{\min}$. Even if $\bm{W}(\boldsymbol{\theta})$ can enact arbitrary global unitaries on $m$ qubits, which requires the number of primitive gates to be exponential in $m$, the fact that $\bm{D}$ is fixed prior to training still restricts the set of functions that can be learned.  These observations have led the community to consider whether there is any benefit in using variational linear quantum models instead of quantum-kernel methods \cite{schuld2021supervised}.

One potential benefit of variational models is that the number of circuit runs used to train the model with  parameter-shift methods \cite{schuld2019evaluating} scales as $\mathcal{O}(\dim(\bm{\theta})\times|\mathcal{T}|)$. For quantum-kernel methods, the complexity of computing the kernel matrix requires $\mathcal{O}(|\mathcal{T}|^2)$ evaluations of $k_Q$. Thus, if the variational optimization of $\bm{\theta}$ converges to an acceptable empirical risk value quickly enough, and  $\dim(\bm{\theta}) \ll |\mathcal{T}|$, then the variational model can have an advantage over the quantum-kernel method. However, the overparameterization of quantum models, i.e. the case where $\dim(\bm{\theta}) \gg |\mathcal{T}|$, has also been investigated \cite{larocca2021theory, liu2021representation, shirai2021quantum, You_2022}.  Additionally, there are forms of regularization applicable to variational linear quantum models for which there is, currently, no analogue for quantum-kernel methods \cite{gyurik2021structural}.

Jerbi \etal~\cite{jerbi2021quantum} proved that VQML models, including data re-uploading ones, can be approximately reduced to variational linear quantum models that use additional ancillas and a quantum embedding whose kernel is the identity. This kernel is classically computable, and furthermore, a quantum-kernel method using the identity matrix as the kernel would simply overfit the data. The authors performed experiments demonstrating that VQML models, including variational linear quantum models, can still generalize better than quantum-kernel methods. This includes cases where regularization, $J$, was applied to the quantum-kernel method used. Thus, it appears that the connection between quantum-kernel methods and variational linear quantum models through the RKHS framework is limited. The authors of \cite{jerbi2021quantum} noted that the generalization advantage that VQML models can have  is based on the fact that, as mentioned, $\bm{O}_{\bm{\theta}}$ restricts the  space of functions that the variational model can represent. Thus, $\bm{O}_{\bm{\theta}}$ not being able to realize arbitrary observables may actually be advantageous and act as additional regularization. Identifying these benefits is important because we will be heavily focusing on variational linear quantum models throughout our work. In section~\ref{sec:embeddings}, we discuss quantum embeddings for encoding Boolean inputs, i.e. $\mathcal{A} := \mathbb{B}^n$, into variational linear quantum models.

\section{Quantum Embeddings for the Boolean Cube}
\label{sec:embeddings}
This section presents the main theoretical contributions of this article. We discuss two quantum embeddings for real-valued functions on the Boolean cube: the \emph{phase embedding} (Section \ref{sec:phase}) and the \emph{QRAC embedding} (Section \ref{sec:qrac}). 
As mentioned in Section \ref{sec:quantumcircuitlearning}, the nonlinearity in the input $\bm{x}$ for models like Equation \eqref{eqn:classical_linear_model}, both classical and quantum, comes from the feature map used. Thus, when analyzing the expressivity of a variational linear quantum model \eqref{eqn:variational_linear_quantum_model}, we will fix the embedding scheme $\rho(\bm{b})$, and determine the class of functions on the Boolean cube that the model can represent. However, when proving theorems, we will assume $\bm{W}(\bm{\theta})$ is universal so that it can enact arbitrary global unitaries on $m$ qubits, which implies that for any $m$-qubit unitary $\bm{V}$ there exists $\bm{\theta}$ such that $\bm{W}(\bm{\theta}) = \bm{V}$. In addition, this means $\bm{W}(\bm{\theta})$ may decompose into a number of  primitive gates that is exponential in $m$. This implies that for any $m$-qubit observable $\bm{M}$, there exists a setting of the parameters $\bm{\theta}$ and a diagonal observable $\bm{D}$ such that

\begin{align}
\bm{O}_{\bm{\theta}} = \bm{W}^{\dagger}(\bm{\theta})\bm{D}\bm{W}(\bm{\theta}) = \bm{M}.
\end{align}
Based on the assumptions just mentioned, if for $g \in \mathcal{G} := \{g: \mathbb{B}^n \rightarrow \mathbb{R}\}$ we can show the existence of an observable $\bm{M}^{(g)}$ such that $\forall \bm{b} \in \mathbb{B}^n: \Tr[\bm{M}^{(g)}\rho(\bm{b})] = g(\bm{b})$, then there exists a variational linear quantum model using $\bm{O}_{\bm{\theta}}^{(g)} = \bm{W}^{\dagger}(\bm{\theta})\bm{D}^{(g)}\bm{W}(\bm{\theta})$ such that $\forall \bm{b} \in \mathbb{B}^n$ one has $\Tr[\bm{O}_{\bm{\theta}}^{(g)}\rho(\bm{b})] = g(\bm{b})$, for some $\bm{D}^{(g)}$ and parameter setting $\bm{\theta}$. In Theorem \ref{thm:phase}, which applies to the phase embedding, we show the existence of such an observable $\bm{M}^{(g)}$ for all functions $g$ on the Boolean cube. The phase embedding produces an $n$-qubit product state for an $n$-bit input. For the QRAC embedding, the observable  $\bm{M}^{(g)}$ exists for a subclass of functions $g$ with degree $\leq \lceil \frac{n}{3} \rceil$. Unfortunately, the Fourier transform of a function in this subclass cannot be nonzero on all elements of $\mathbb{B}^{n}$ with Hamming weight $\leq \lceil \frac{n}{3} \rceil$. The exact conditions are presented in Theorems \ref{thm:qrac} and \ref{thm:permuted_spectra}. The product state $\rho_{\text{QE}}(\bm{b})$ consists of $\lceil\frac{n}{3}\rceil$ qubits for an $n$-bit input vector.

We also present sufficient conditions under which a classical ensemble of VQML models can express functions on the Boolean cube. More specifically, we call the summation of the outputs of multiple variational linear quantum models a \textit{classical ensemble of quantum models}, i.e.
\begin{equation}
\label{eqn:ensemble}
    \sum_{i \in \mathcal{D}}f^{(i)}_{\bm{\theta}_i}(\bm{b}),
\end{equation}
where each $f^{(i)}_{\bm{\theta}_i}$, indexed over a set $\mathcal{D}$, uses the same embedding scheme. Prior work has dealt with quantum ensembles, i.e. superpositions, of models \cite{schuld2018quantum,abbas2020quantum,araujo2020quantum,macaluso2020quantum}.
Our results on ensembles show the existence of a collection of observables $\{\bm{M}^{(g, i)}\}_{i \in \mathcal{D}}$ indexed from some set $\mathcal{D}$ such that $\forall \bm{b} \in \mathbb{B}^n$ one has
\begin{align}
\sum_{i \in \mathcal{D}} \Tr[\bm{M}^{(g, i)}\rho(\phi_{i}(\bm{b}))] = g(\bm{b}),
\end{align}
where each $\phi_i$ classically preprocesses the input bits. The potential impact that classical preprocessing of the input data can have on VQML models was acknowledged in \cite{Havl_ek_2019, 2020reupload, Schuld_2021, thumwanit2021trainable}. Because the ensemble is the sum of the outputs of the multiple linear quantum models, if analytic-gradient learning is used, then the parameters of the models can be updated in parallel.
For the phase embedding an ensemble of $\binom{n}{d}$ models is sufficient to express any function of $n$ input Boolean variables and degree at most $d$, see Theorem \ref{thm:phase_esemb}. The preprocessing functions select subsets of the input variables.
In the case of the QRAC embedding, according to Theorem \ref{thm:qrac_ensemb}, the class of expressible functions is all of those with degree $\leq \lceil \frac{n}{3}\rceil$ with the preprocessing functions being permutations, i.e. elements of the symmetric group on $n$ elements.

As mentioned earlier, data re-uploading models, can be converted into variational linear quantum models \cite{jerbi2021quantum}. This can be done approximately by introducing additional quantum registers to encode the gate parameters and additional controlled-rotation gates. There exist transformations that are exact, but they require either gate teleportation, which introduces additional classically-controlled rotation gates that are dependent on the input, or post selection. However, for the Boolean cube we show that standard variational linear quantum models are sufficient, and thus the mentioned transformations are not required. Although repeating the phase or QRAC embeddings sequentially, even without inserting trainable gates between repetitions, can provide some benefits, see Appendix \ref{sec:repeated_embeddings}.

In practice, to construct a model that can be efficiently implemented, we need to select $\bm{W}(\bm{\theta})$ such that it decomposes into $\mathcal{O}(\text{poly}(m))$ primitive gates and select $\bm{D}$ to be a linear combination of $\mathcal{O}(\text{poly}(m))$ elements from  $\{\mathbb{I}, \mathsf{Z}\}^{\otimes m}$. %
Such choices will introduce regularization, as mentioned in Section \ref{sec:quantumcircuitlearning}, and restrict the class of functions in the RKHS that can be represented. The goal of variational optimization will be to find a setting of the parameters, $\bm{\theta}$, if it exists, such that $\bm{O}^{(g)}_{\bm{\theta}} = \bm{M}^{(g)}$. In addition, even if such a choice of parameters exists, the ability to find $\bm{\theta}$ through variational optimization will also depend on the loss function landscape. The loss landscape for variational models has been observed to be difficult to navigate in practice when the PQC used is highly expressive \cite{mcclean2018barren, Bittel_2021}. The goal of the experiments in Section \ref{sec:experiments} is to demonstrate two cases in which the optimization is possible.

Lastly, we make a comment on the related work of Thumwanit \etal~\cite{thumwanit2021trainable}. The authors showed that Pauli rotations can be used encode discrete-valued inputs using fewer qubits than the total number of input bits by making use of a classical preprocessing function that maps tuples of input bits to trainable rotation angles. This introduces additional trainable parameters that are not present in the phase and QRAC embeddings. Also, it is not guaranteed that there always exists a mapping of multiple input bits to rotation angles that is sufficient for expressing the target real-valued function. 

\subsection{Phase Embedding}
\label{sec:phase}
The phase embedding that we investigate was considered in Schuld \etal~\cite{Schuld_2021} for continuous-valued inputs. They showed that if  the phase embedding is repeated $r$ times sequentially with trainable blocks in between repetitions or repeated $r$ times in parallel, then the output of the variational model is expressible as the $r^\text{th}$ cubic partial sum \cite{weisz2012summability} of a Fourier series:
    \begin{equation}
        \label{eqn:partial_fourier}
        \sum_{\substack{\bm{\omega} \in \mathbb{Z}^{n} : \lVert\bm{\omega}\rVert_{\infty} \leq r}} c_{\bm{\omega}}e^{i\bm{\omega}\cdot\bm{x}},
    \end{equation}
where $\bm{x}$ is the $n$-dimensional input.
The trainable blocks and observable control the $c_{\bm{\omega}}$'s. These models can arbitrarily approximate functions in $L_2([0, 2\pi]^n)$ if we assume the freedom to choose arbitrarily deep trainable blocks and an arbitrary observable, and utilize arbitrarily many repetitions of the embedding layer. More specifically, $\forall\epsilon > 0$ and any function in $L_2([0, 2\pi]^n)$  there exists a model using some number $r$ of repetitions of the phase embedding that approximates that function to error $\epsilon$ in $L_2$ norm. In this section, we show that any function on the Boolean cube can be represented by a variational linear quantum model, i.e. the case where $r=1$, that uses the phase embedding.

Let $\mathsf{H}$ be the Hadamard operator, and
\begin{align}
    \mathsf{X}^{(\bm{b})} := \bigotimes_{i=1}^n \mathsf{X}^{b_i},
\end{align}
where $\mathsf{X}$ is the Pauli-$\mathsf{X}$ operator.
For $\bm{b} \in \mathbb{B}^n$, we define the following to be the \emph{phase embedding unitary}:
\begin{equation}
\label{eqn:phaseembedunitary}
    \bm{U}_{\text{PE}}(\bm{b}) := \mathsf{H}^{\otimes n}\mathsf{X}^{(\bm{b})}.
\end{equation}
Moreover,
\begin{equation}
    \rho_{\text{PE}}(\bm{b}) := \bm{U}_{\text{PE}}(\bm{b})\ketbra{0_n}\bm{U}_{\text{PE}}^{\dagger}(\bm{b})
\end{equation}
is the associated feature state. The following theorem summarizes the expressiveness of variational linear quantum models that makes use of this embedding.

\begin{theorem} 
\label{thm:phase}

For any $g \in \mathcal{G}$, there exists an observable $\bm{O}$ such that $\forall \bm{b} \in \mathbb{B}^n$ one has $\Tr[\bm{O}\rho_{\operatorname{PE}}(\bm{b})] = g(\bm{b})$.

\end{theorem}

\begin{proof}
Suppose $\bm{O}$ is an observable such that in the computational basis its entries are $\bm{O}_{\bm{k}, \bm{j}} = \widehat{g}(\bm{k}\oplus\bm{j})$. Then,
\begin{align}
\begin{split}
    \label{eqn:phase}
     \Tr[\bm{O}\rho_{\text{PE}}(\bm{b})] &= \bra{0_n}\mathsf{X}^{(\bm{b})}\mathsf{H}^{\otimes n}\bm{O}\mathsf{H}^{\otimes n}\mathsf{X}^{(\bm{b})}\ket{0_n}
    \\&= \frac{1}{2^n}\sum_{\bm{k} \in \mathbb{B}^n}\sum_{\bm{j} \in \mathbb{B}^n}(-1)^{\bm{b}\cdot(\bm{k}\oplus \bm{j})}\bm{O}_{\bm{k},\bm{j}} \\
    &= \sum\limits_{\bm{s} \in \mathbb{B}^n}\widehat{g}(\bm{s})\chi_{\bm{s}}(\bm{b}) 
    \\&(\text{since}~ |\{(\bm{k}, \bm{j}) \in \mathbb{B}^n \times  \mathbb{B}^n~| ~\bm{k} \oplus \bm{j} = \bm{s}\}| = 2^n)
    \\&= g(\bm{b}),
\end{split}
\end{align}
as required.
\end{proof}
Additionally, $\bm{O}$ can always be diagonalized: $\bm{O} = \bm{V}\bm{D}\bm{V}^{\dagger}$ and implemented by a variational linear quantum model with universal trainable gates, i.e. $\bm{W}(\bm{\theta}) = \bm{V}^{\dagger}$. The model uses a measurement $\bm{D} \in \textup{span}_\mathbb{R}\{ \mathbb{I}, \mathsf{Z}\}^{\otimes m}$. Each diagonal Pauli tensor in the linear combination that represents $\bm{D}$ can be measured using separate circuit runs, and the expectation value is then scaled by the corresponding coefficient of that diagonal Pauli tensor. In practice, $\bm{D}$ is typically chosen to be a simple and easy to evaluate observable, such as $\mathsf{Z}^{\otimes m}$, instead of being chosen arbitrarily from $\textup{span}_\mathbb{R}\{ \mathbb{I}, \mathsf{Z}\}^{\otimes m}$ as done in the proof of Theorem \ref{thm:phase}.  Regardless,  if $\bm{W}(\bm{\theta})$ can implement arbitrary unitary operators, the model can only implement observables that have the same spectrum as $\bm{D}$, i.e. $\lVert g\rVert_{\infty}$ remains bounded.
We leave as an open question as to whether interesting classes of functions exist that can be expressed when the spectrum of $\bm{O}_{\bm{\theta}}$ is fixed, i.e. fixed $\bm{D}$ but varying $\bm{W}(\bm{\theta})$.

If we have prior knowledge that the function is a $k$-junta, we can potentially reduce the number of qubits that the trainable portion of the model acts on utilizing a variational $\mathsf{SWAP}$ network. This is discussed in more detail in Appendix \ref{sec:phase_swaps}.

The name ``phase embedding'' comes from the fact that it maps the  input bits to $\ket{\pm}$, i.e. $\mathsf{X}$-basis states. Equivalently, the phase embedding unitary can be replaced by $\mathsf{Z}^{(\bm{b})}\mathsf{H}^{\otimes n}$, where $\mathsf{Z}$ is the Pauli-$\mathsf{Z}$ operator. The phase embedding for continuous inputs, described in \cite{Schuld_2021}, was of the form:
\begin{equation}
\label{eqn:continuous_phase_embed}
     \bigotimes_{i=1}^{n}(\bm{R}_{\mathsf{Z}}(x_i)\mathsf{H}),
\end{equation}
where $\bm{R}_{\mathsf{Z}}$ is a rotation generated by $\mathsf{Z}$, and $\bm{x} \in [0, 2\pi]^n$. More specifically, this is the operator $\bm{R}_{\mathsf{Z}}(\theta) := e^{-i\frac{\theta}{2}\mathsf{Z}}$.

The connection to Equation \eqref{eqn:phaseembedunitary} is made by using the restriction of Equation \eqref{eqn:continuous_phase_embed} to the set $\{0, \pi\}^n$, which can be seen to be equal to
\begin{equation}
\label{eqn:z_phase_embed}
  \bigotimes_{i=1}^{n}(\mathsf{Z}^{b_i}\mathsf{H}) = \mathsf{Z}^{(\bm{b})}\mathsf{H}^{\otimes n},
\end{equation}
up to a global phase.
Based on the proof given above, we define the Fourier coefficients of a variational linear quantum model, Equation \eqref{eqn:variational_linear_quantum_model}, that uses the phase embedding to be
\begin{equation}
\label{eqn:phase_fourier_coeffs}
\widehat{f}_{\bm{\theta}}(\bm{s}) = \frac{1}{2^n}\sumtwo_{\substack{~\bm{k} \in \mathbb{B}^n~\bm{j} \in \mathbb{B}^{n} \\ \bm{k}\oplus\bm{j} = \bm{s}}}(\bm{O}_{\bm{\theta}})_{\bm{k},\bm{j}}.
\end{equation}
This definition will be used in Section \ref{sec:experiments} to see how well the model was able to fit the Fourier coefficients of the target function using supervised learning. 

We now present sufficient conditions under which an ensemble of models utilizing the phase embedding can represent functions on $n$ bits utilizing fewer than $n$ qubits. Let $\text{wt}(\cdot)$ compute the Hamming weight of a binary vector. Additionally, for each $d \in [n]$, consider the set
\begin{equation}
    \mathcal{C}_{n, d} := \{ (i_1, i_2, \dots, i_d) \in [n]^{d}~|~i_1 < i_2 < \cdots < i_d\},
\end{equation}
which contains an ordered $d$-tuple, for each subset of $[n]$ of size $d$. Thus, it follows that $\lvert \mathcal{C}_{n,d} \rvert = \binom{n}{d}$. For each $\bm{w} \in \mathcal{C}_{n,d}$, we define a function $\nu_{\bm{w}} : \mathbb{B}^{n} \xrightarrow[]{} \mathbb{B}^{d}$ such that $\nu_{\bm{w}}(\bm{b}) = (b_{w_1}, b_{w_2}, \dots, b_{w_d})$, i.e. selects a subset of the entries of $\bm{b}$ indicated by $\bm{w}$. 

\begin{theorem}
\label{thm:phase_esemb}
For any $g \in \mathcal{G}$ with degree $d \in [n]$ there exists a collection of observables $\{\bm{O}^{(\bm{w})}\}_{\bm{w} \in \mathcal{C}_{n,d}}$ such that $\forall \bm{b} \in \mathbb{B}^n$ one has $\sum_{\bm{w} \in \mathcal{C}_{n,d}}\Tr[\bm{O}^{(\bm{w})}\rho_{\operatorname{PE}}(\nu_{\bm{w}}(\bm{b}))] = g(\bm{b})$, where $\rho_{\operatorname{PE}}(\bm{b})$ is a $d$-qubit state.
\end{theorem}

\begin{proof}
Let $\eta_{\bm{w}}:\mathbb{B}^{d} \xrightarrow[]{} \mathbb{B}^{n}$ be the right inverse of $\nu_{\bm{w}}$ that sets all entries with indices not in $\bm{w}$ to zero, i.e. $\nu_{\bm{w}} \circ \eta_{\bm{b}}$ is the identity function on $\mathbb{B}^{d}$. Let $k_{\bm{s}} := \binom{n}{n -d + \text{wt}(\bm{s})}$ for all $\bm{s} \in \{ \bm{b} \in \mathbb{B}^n~|~\text{wt}(\bm{b}) \leq d\}$. For each $\bm{w} \in \mathcal{C}_{n, d}$, let $\widehat{\bm{w}} := \{\bm{s} \in \mathbb{B}^n | s_i = 1 \implies i \in \bm{w}\}$ and define the function $g^{(\bm{w})}:\mathbb{B}^{d} \xrightarrow[]{} \mathbb{R}$ as follows

\begin{equation}
    \label{eqn:component_g}
    g^{(\bm{w})}(\bm{b}) := \sum_{\bm{s} \in \widehat{\bm{w}}}\frac{\widehat{g}(\bm{s})}{k_{\bm{s}}}\chi_{\bm{s}}(\eta_{\bm{w}}(\bm{b})).
\end{equation}
Note that for all $\chi_{\bm{s}}$ in \eqref{eqn:component_g} $\chi_{\bm{s}}(\eta_{\bm{w}}\circ \nu_{\bm{w}}(\bm{b})) = \chi_{\bm{s}}(\bm{b})$ since $g^{(\bm{w})}$ only depends on the input bits indexed by $\bm{w}$.
Thus, $\forall \bm{b} \in \mathbb{B}^n$
\begin{equation}
    g(\bm{b}) = \sum_{\bm{w} \in \mathcal{C}_{n, d}}g^{(\bm{w})}(\nu_{\bm{w}}(\bm{b})),
\end{equation}
since $\frac{\widehat{g}(\bm{s})}{k_{\bm{s}}}\chi_{\bm{s}}$ appears $k_{\bm{s}}$ times in the sum. By Theorem \ref{thm:phase}, for each $g^{(\bm{w})}$ there exists an observable $\bm{O}^{(\bm{w})}$ such that $\forall \bm{b}' \in \mathbb{B}^{d}$, $\Tr[\bm{O}^{(\bm{w})}\rho_{\operatorname{PE}}(\bm{b}')] = g^{(\bm{w})}(\bm{b}')$. Thus, it follows that $\forall \bm{b} \in \mathbb{B}^n$, 
\begin{equation}
    \sum_{\bm{w} \in \mathcal{C}_{n,d}}\Tr[\bm{O}^{(\bm{w})}\rho_{\operatorname{PE}}(\nu_{\bm{w}}(\bm{b}))]  = g(\bm{b}), 
\end{equation}
as required.
\end{proof}
This proof shows that the upper bound on the number of models in the ensemble, using the phase embedding, is $\binom{n}{d}$ for a degree $d$ function. When training, we can utilize a validation data set to determine if the size of the ensemble is sufficient. However, restricting the size of the ensemble also acts as a regularizer.

Alternatively, we can make use of repetitions of the embedding to increase the expressivity of a single model that uses the phase embedding. Consider the following $m$-qubit embedding with $r$ data-encoding steps:

\begin{equation}
    \prod_{j=1}^{r}\mathsf{Z}^{(\nu_{\bm{w}_j}(\bm{b}))}\bm{V}_{\bm{w}_j},
\end{equation}

where $\{\bm{w}_j\} \subset \mathcal{C}_{n,m}$ for $m < n$. Note we have used the equivalent representation of the phase embedding presented in \eqref{eqn:z_phase_embed}. The unitary operators $V_{\bm{w}_j}$ could be fixed or trainable. If they are fixed, then the overall model is still a linear quantum model, i.e. not a data re-uploading model, as all embedding steps come before the trainable components. However, we ensure that $V_{\bm{w}_1} = \mathsf{H}^{\otimes m}$. In Appendix \ref{sec:phase_improving_through_rep}, we demonstrate that a model using this embedding can be expressed as a non-trivial linear combination of all Fourier basis terms when $r=n/m$ and $m$ divides $n$. Thus it can express degree-$n$ functions on $\mathbb{B}^n$ with fewer than $n$ qubits. 

In what follows, we discuss a different scheme that encodes multiple bits into a single qubit.

\subsection{QRAC Embedding}
\label{sec:qrac}
In this section, we analyze a quantum embedding, first described by Yano \etal~\cite{yano2020efficient}, based on QRACs. However, the authors did not perform any analysis of the expressivity of models that make use of it. QRACs have also been used to encode MaxCut problems solved by variational quantum optimization algorithms \cite{fuller2021approximate}. We develop sufficient conditions for variational linear quantum models using the QRAC embedding to be able to express functions on the Boolean cube. Moreover, we show that an ensemble of models each using this embedding and $\lceil \frac{n}{3} \rceil$ qubits is  sufficient to represent any function on $\mathbb{B}^n$ with degree $d \leq \lceil \frac{n}{3} \rceil$. While this embedding can only provide a constant-factor saving in terms of the number of qubits, it could be impactful during the era of small and noisy quantum hardware provided that efficient and useful PQCs can be constructed. Like for the phase embedding in Theorem \ref{thm:phase_esemb}, the qubit reduction provided by the QRAC embedding comes from a set of classical preprocessing functions, described below.  We use the (3,1)--QRAC, which is a three-bits-to-one-qubit probabilistic encoding scheme. It was introduced by Ike Chuang and first mentioned in \cite{ambainisqrac}. Thus, the number of (3,1)--QRACs used to encode an $n$-bit input is $\lceil \frac{n}{3}\rceil$, which requires padding with passive variables if the input $\bm{b}$ is a tuple whose length is not divisible by 3. Note that this only increases the input length by at most two more zero bits and does not change the number of qubits used. Due to this, we assume, without loss of generality, $n \equiv 0$ mod 3.

We start by dividing the input $\bm{b}$ into $\frac{n}{3}$ triplets $\mathcal{B}_1, \mathcal{B}_2, \dots, \mathcal{B}_{\frac{n}{3}}$, each of which is an element of $\mathbb{B}^3$. The entries of the $i^\text{th}$ triplet, $\mathcal{B}_i$, are indexed by the symbols for Pauli operators: $\mathcal{B}_i^{(\mathsf{X})} = b_{3i-2}$, $\mathcal{B}_i^{(\mathsf{Y})} = b_{3i-1}$, and $\mathcal{B}_i^{(\mathsf{Z})} = b_{3i}$. In addition, we define $\mathcal{B}_i^{(\mathbb{I})} := 0$. For some angles $\alpha_1$ and $\alpha_2$, the $i^\text{th}$ triplet can be encoded in the following single-qubit state:
\begin{align}
\label{eqn:gen_qrac}
&\sigma(\mathcal{B}_i)_{\alpha_1, \alpha_2} := \frac{1}{2} \Bigg(\mathbb{I} + (\sin(\alpha_2)\cos(\alpha_1))(-1)^{\mathcal{B}_i^{(\mathsf{X})}}\mathsf{X} \nonumber\\&+ (\sin(\alpha_2)\sin(\alpha_1))(-1)^{\mathcal{B}_i^{(\mathsf{Y})}}\mathsf{Y} + \cos(\alpha_2)(-1)^{\mathcal{B}_i^{(\mathsf{Z})}}\mathsf{Z} \Bigg).
\end{align} 
The unitary that takes $\ketbra{0}$ to $\sigma(\mathcal{B}_i)_{\alpha_1, \alpha_2}$ can be expressed as a composition of two rotations:
\begin{equation}
\label{eqn:qrac_rotations}
\bm{R}_{\mathsf{Z}}(\phi_{\mathsf{Z}}(\alpha_1, \mathcal{B}_i^{(\mathsf{X})},\mathcal{B}_i^{(\mathsf{Y})}))\bm{R}_{\mathsf{Y}}(\phi_{\mathsf{Y}}(\alpha_2, \mathcal{B}_i^{(\mathsf{Z})})),
\end{equation}
where 
\begin{equation}
\begin{split}
    \phi_{\mathsf{Z}}(\alpha_1, \mathcal{B}_i^{(\mathsf{X})},\mathcal{B}_i^{(\mathsf{Y})}) &= (\mathcal{B}_i^{(\mathsf{X})} (\pi - \alpha_1) + (1-\mathcal{B}_i^{(\mathsf{X})})\alpha_1)(-1)^{\mathcal{B}_i^{(\mathsf{Y})}}, \\ 
    \phi_{\mathsf{Y}}(\alpha_2, \mathcal{B}_i^{(\mathsf{Z})}) &= \mathcal{B}_i^{(\mathsf{Z})} (\pi - \alpha_2) + (1-\mathcal{B}_i^{(\mathsf{Z})})\alpha_2.
\end{split}
\end{equation}
Specifically, when
$\alpha_1 = \frac{\pi}{4}$ and $\alpha_2 = 2\cos^{-1}\left(\sqrt{\frac{1}{2} + \frac{1}{2\sqrt{3}}}\right)$ we obtain the (3,1)--QRAC state \cite{ALMO08}:
\begin{align}
\label{eqn:single_qrac_state}
&\sigma(\mathcal{B}_i) := \bm{U}_{3,1}(\mathcal{B}_i)\ketbra{0}\bm{U}_{3,1}^{\dagger}(\mathcal{B}_i) \\&= \frac{1}{2} \left(\mathbb{I} + \frac{(-1)^{\mathcal{B}_i^{(\mathsf{X})}}}{\sqrt{3}}\mathsf{X} + \frac{(-1)^{\mathcal{B}_i^{(\mathsf{Y})}}}{\sqrt{3}}\mathsf{Y} + \frac{(-1)^{\mathcal{B}_i^{(\mathsf{Z})}}}{\sqrt{3}}\mathsf{Z} \right),\nonumber
\end{align}
where $\bm{U}_{3,1}(\mathcal{B}_i)$ is the operation in Equation \eqref{eqn:qrac_rotations} for the specific assignments of $\alpha_1$ and $\alpha_2$ mentioned above. These choices for $\alpha_1$ and $ \alpha_2$ maximize the probability of recovering a single bit when measuring along one of the three Bloch-sphere axes. While we will be making use of the state in Equation \eqref{eqn:single_qrac_state}, in our case, all of the results that follow would still hold if we had utilized any state with the form presented in Equation \eqref{eqn:gen_qrac} and  different $\alpha_1$ and $\alpha_2$, as long as all coefficients of the Pauli operators are nonzero. The reason for this is that the proofs that follow only depend on the relationships between $\mathcal{B}_i^{(\mathsf{X})}, \mathcal{B}_i^{(\mathsf{Y})}, \mathcal{B}_i^{(\mathsf{Z})}$ and the powers of $(-1)$ that appear in the coefficients of the Pauli terms.

Generalizing to $n > 3$, we can encode $\bm{b} \in \mathbb{B}^n$ by using $\frac{n}{3}$ qubits in the following product state: 
\begin{align}
    \label{eqn:qracstate}
    &\rho_{\text{QE}}(\bm{b}) := \bigotimes_{i=1}^{n/3} \sigma(\mathcal{B}_i) \\&= \frac{1}{2^{n/3}}\bigotimes_{i=1}^{n/3} \left(\mathbb{I} + \frac{(-1)^{\mathcal{B}_i^{(\mathsf{X})}}}{\sqrt{3}}\mathsf{X} + \frac{(-1)^{\mathcal{B}_i^{(\mathsf{Y})}}}{\sqrt{3}}\mathsf{Y} + \frac{(-1)^{\mathcal{B}_i^{(\mathsf{Z})}}}{\sqrt{3}}\mathsf{Z} \right).\nonumber
\end{align} Thus, the \emph{QRAC embedding unitary} is
\begin{equation}
    \label{eqn:qrac_embedding_unitary}
    \bm{U}_{\text{QE}}(\bm{b}) := \bigotimes_{i=1}^{n/3}\bm{U}_{3,1}(\mathcal{B}_i).
\end{equation}
Lastly, $\forall m \in \mathbb{N}$, we define the sets 
\begin{equation}
    \mathcal{K}^{\operatorname{QE}}_{m} = \{ \bm{s} \in \mathbb{B}^{3m}~|~\forall i \in [m]: \text{wt}((s_{3i-2},s_{3i-1}, s_{3i}))\leq 1\},
\end{equation}
where $\text{wt}(\cdot)$ computes the Hamming weight of a binary vector. For a fixed $m$, $\mathcal{K}_{m}^{\operatorname{QE}}$ contains all elements, $\bm{s}$, of $\mathbb{B}^{3m}$ with Hamming weight $m$ such that $\forall i \in [\frac{n}{3}]$, the $i$-th triplet, $\mathcal{B}_i$, of $\bm{s}$ has Hamming weight at most one. These sets
will play a role in the results that follow.

The first result of Section \ref{sec:qrac} is Theorem \ref{thm:qrac}. This result presents sufficient conditions for a function on the Boolean cube to be expressible by variational linear quantum models that use the QRAC embedding.
\begin{theorem}[]
\label{thm:qrac}
For any $g \in \mathcal{G}$ such that  $\widehat{g}(\bm{s}) \neq 0 \implies \bm{s} \in \mathcal{K}^{\operatorname{QE}}_{\frac{n}{3}}$, there exists an observable $\bm{O}$ such that $\forall \bm{b} \in \mathbb{B}^n$ one has  $\Tr[\bm{O}\rho_{\operatorname{QE}}(\bm{b})] = g(\bm{b})$.
\end{theorem}
\begin{proof}
Let $m = \frac{n}{3}$. We start by expanding the tensor product in the definition of the quantum state in Equation \eqref{eqn:qracstate}:
\begin{equation}
\label{eqn:qracstatesum}
    \rho_{\text{QE}}(\bm{b}) = \sum_{\bm{P} \in \{\mathbb{I}, \mathsf{X}, \mathsf{Y}, \mathsf{Z}\}^{\otimes m}} \frac{\chi_{\bm{P}}(\bm{b})}{2^{m}3^{|\bm{P}|/2}}\bm{P},
\end{equation}
where $|\bm{P}|$ is the number of non-identity Pauli operators in the simple tensor $\bm{P}$ and  
\begin{align}
    \chi_{\bm{P}}(\bm{b}) := \prod_{i=1}^{m} (-1)^{\mathcal{B}_i^{(\bm{P}_i)}}.
\end{align}

Consider $\Phi : \{\mathbb{I}, \mathsf{X}, \mathsf{Y}, \mathsf{Z} \}^{\otimes m} \xrightarrow[]{} \mathcal{K}^{\text{QE}}_{m}$ defined as follows. For each $i \in [m]$, the bit triplet $(\Phi(\bm{P})_{3i-2},\Phi(\bm{P})_{3i -1},\Phi(\bm{P})_{3i})$ has a $1$ in the first, second, or third position if and only if $\bm{P}_i$ is $\mathsf{X}$, $\mathsf{Y}$, or $\mathsf{Z}$, respectively, and it is a triplet of zeroes otherwise.  We choose $\bm{O}$ in the following way:
\begin{equation}
\label{eqn:proof_observ}
  \bm{O} := \sum_{\bm{P} \in \{\mathbb{I}, \mathsf{X}, \mathsf{Y}, \mathsf{Z}\}^{\otimes m}} 2^{m}3^{|\bm{P}|/2}\widehat{g}(\Phi(\bm{P}))\bm{P}.
\end{equation}
Since $\Phi$ is a bijection by construction,  for each $\bm{s} \in \mathcal{K}^{\text{QE}}_{m}$ we can associate the Fourier basis element $\chi_{\bm{s}}$ with a $\chi_{\bm{P}}$ in Equation \eqref{eqn:qracstatesum}. It follows that
\begin{align}
    \label{eqn:qracfun}
 &\text{Tr}[\bm{O}\rho_{\text{QE}}(\bm{b})] =  \sum_{\bm{P} \in \{\mathbb{I}, \mathsf{X}, \mathsf{Y}, \mathsf{Z}\}^{\otimes m}} \frac{\chi_{\bm{P}}(\bm{b})}{2^{m}3^{|P|/2}}\Tr[\bm{O}\bm{P}]\nonumber \\&= \sum\limits_{\bm{s} \in \mathcal{K}^{\text{QE}}_{m}} \widehat{g}(\bm{s})\chi_{\bm{s}}(\bm{b}) = g(\bm{b}),
\end{align}
as required.
\end{proof}

Similar to Section \ref{sec:phase}, we define the Fourier coefficients of a variational linear quantum model that uses the QRAC embedding to be
\begin{equation}
    \label{eqn:qrac_fourier_coeffs}
    \widehat{f}_{\bm{\theta}}(\bm{s}) = \frac{\Tr[\bm{O}_{\bm{\theta}}\Phi^{-1}(\bm{s})]}{2^{m}3^{\lvert\Phi^{-1}(\bm{s})\rvert/2}}.
\end{equation}
As mentioned in Section \ref{sec:phase}, this definition will be used in Section \ref{sec:experiments} to see how well the variational model was able to fit the Fourier coefficients of the target function. Lastly, similar to the phase embedding case $\bm{O}$ can be diagonalized into $\bm{V}\bm{D}\bm{V}^{\dagger}$, where $\bm{D} \in \textup{span}_\mathbb{R}\{ \mathbb{I}, \mathsf{Z}\}^{\otimes m}$. We again leave as an open question if interesting classes of functions can be expressed when the spectrum of $\bm{O}_{\bm{\theta}}$ is fixed.

Before moving to our next result, we introduce another concept. With respect to an initial ordering of the input variables $(b_1, b_2, \dots, b_n) = \bm{b}$, we define the \emph{$\tau$-permuted model} to be 
\begin{equation}
\label{eqn:permuted_model}
\Tr[\bm{O}_{\bm{\theta}}\rho_{\text{QE}}(\tau(\bm{b}))],
\end{equation}
where $\tau$ is any element of the symmetric group on $n$ elements, $\mathcal{S}_{n}$. 
An element $\tau \in  \mathcal{S}_{n}$ acts on the tuples $\bm{b} \in \mathbb{B}^{n}$ by permuting the order of the entries, where $\tau(\bm{b})$ denotes the permuted tuple. The action of $\tau$ naturally extends to sets, and so it follows that $\tau$ maps $\mathcal{K}_{\frac{n}{3}}^{\text{QE}}$ to
\begin{equation}
    \tau(\mathcal{K}^{\operatorname{QE}}_{\frac{n}{3}}) = \{ \tau(\bm{s})~|~\forall \bm{s} \in \mathcal{K}^{\operatorname{QE}}_{\frac{n}{3}}\}.
\end{equation}
Thus, permuting the input bits expands the class of functions expressible by a single VQML model using the QRAC embedding.

Next, we present Theorem \ref{thm:permuted_spectra} that will be useful for extending the class of functions we can represent with the QRAC embedding and applies to $\tau$-permuted models.
\begin{theorem}
\label{thm:permuted_spectra}
For any $g \in \mathcal{G}$, if  $\exists\tau \in \mathcal{S}_n$ such that   $\forall \bm{s} \in \mathbb{B}^n$ the condition  $\widehat{g}(\bm{s}) \neq 0 \implies \bm{s} \in \tau(\mathcal{K}^{\operatorname{QE}}_{\frac{n}{3}})$ is satisfied, then there exists an observable $\bm{O}$  such that $\forall \bm{b} \in \mathbb{B}^{n}$ one has $\Tr[\bm{O}\rho_{\operatorname{QE}}(\tau(\bm{b}))] = g(\bm{b})$.
\end{theorem}
\begin{proof}
The main argument is based on the simple fact that
\begin{equation}
\label{eqn:permutation_swap}
\forall \bm{s}, \bm{b} \in \mathbb{B}^n, \bm{s} \cdot \tau(\bm{b}) = \tau(\bm{s}) \cdot \bm{b}.
\end{equation} Let $\Phi$ be as defined in the proof of Theorem \ref{thm:qrac}, and suppose for $\bm{P} \in \{\mathbb{I}, \mathsf{X}, \mathsf{Y}, \mathsf{Z} \}^{\otimes m}$, $\bm{s} = \Phi(\bm{P})$ and $\bm{P}'$ is such that $\Phi(\bm{P}') = \tau(\Phi(\bm{P}))$. This $\bm{P}'$ exists because $\Phi$ is bijective. Then, using Equation \eqref{eqn:permutation_swap}, it follows that
\begin{equation}
\chi_{\bm{P}}(\tau(\bm{b})) = \chi_{\bm{P}'}(\bm{b}) = \chi_{\tau(\bm{s})}(\bm{b}).
\end{equation} If we replace $\widehat{g}(\Phi(\bm{P}))\bm{P}$ with $\widehat{g}(\tau(\Phi(\bm{P})))\bm{P}$ in Equation \eqref{eqn:proof_observ}, then the rest follows by using the same arguments made when proving Theorem \ref{thm:qrac}.
\end{proof}

The class of functions that can be represented by ensembles of models, Equation \eqref{eqn:ensemble}, using the QRAC embedding is summarized in the following result. The proof makes use of techniques that are similar to those used in proving Theorem \ref{thm:phase_esemb}.

\begin{theorem}
\label{thm:qrac_ensemb}
For any $g \in \mathcal{G}$ with degree $d \leq \lceil\frac{n}{3}\rceil$, there exists a collection of observables $\{\bm{O}^{(\tau)}\}_{\tau \in \mathcal{S}_n}$ such that $\forall \bm{b} \in \mathbb{B}^n$ one has $\sum_{\tau \in \mathcal{S}_n}\Tr[\bm{O}^{(\tau)}\rho_{\operatorname{QE}}(\tau(\bm{b}))] = g(\bm{b})$.
\end{theorem}
\begin{proof}
Let $m=\frac{n}{3}$, by hypothesis, $g$ satisfies:
\begin{equation}
    \widehat{g}(\bm{s}) \neq 0 \implies \bm{s} \in \{ \bm{b} \in \mathbb{B}^{3m}~|~\text{wt}(\bm{b})\leq m\}.
\end{equation}
For any $\bm{s}$ such that $\text{wt}(\bm{s}) \leq m$, let $k_{\bm{s}}$ be the number of $\tau \in \mathcal{S}_{n}$ such that $\bm{s} \in \tau(\mathcal{K}_{m}^{\text{QE}})$. It can be easily seen that $k_{\bm{s}} \neq 0$ for all such $\bm{s}$ because
\begin{equation}
     \bigcup_{\tau \in \mathcal{S}_{n}} \tau(\mathcal{K}^{\text{QE}}_{m}) = \{ \bm{b} \in \mathbb{B}^{3m}~|~\text{wt}(\bm{b})\leq m\}.
\end{equation}
It is possible that for two different $\tau, \lambda \in \mathcal{S}_{n}$
\begin{align}
    (\tau(\mathcal{K}^{\text{QE}}_{m}) \cap \lambda(\mathcal{K}^{\text{QE}}_{m}))\setminus\{\bm{0}\} \neq \varnothing,
\end{align}
where $\bm{0}$ is the $n$-tuple with all zero entries.
One reason is that $\mathcal{K}^{\text{QE}}_{m}$ has a nontrivial stabilizer group under the action of $\mathcal{S}_{n}$. For example, a permutation that just changes the order of $s_{3i-2},s_{3i-1},s_{3i}$ for some $i \in [m]$ and all $\bm{s} \in \mathcal{K}^{\text{QE}}_{m}$ is a nontrivial stabilizer. Thus multiple permuted models can effectively be identical, i.e. $\mathcal{K}^{\text{QE}}_{m}$ can equal $\tau(\mathcal{K}^{\text{QE}}_{m})$. However, the proof still works if we do not exclude such cases. 

Next, for every $\tau \in \mathcal{S}_n$, we define a new function $g^{(\tau)}:\mathbb{B}^{n} \xrightarrow[]{} \mathbb{R}$ as follows:
\begin{equation}
    g^{(\tau)}(\bm{b}) := \sum_{\bm{s} \in \tau(\mathcal{K}_{m}^{\text{QE}})} \frac{\widehat{g}(\bm{s})}{k_{\bm{s}}}\chi_{\bm{s}}(\bm{b}).
\end{equation}
By invoking Theorem \ref{thm:permuted_spectra}, for each $\tau \in \mathcal{S}_n$, there exists an observable $\bm{O}^{(\tau)}$ such that $\forall \bm{b} \in \mathbb{B}^n$ one has that $\Tr[\bm{O}^{(\tau)}\rho_{\operatorname{QE}}(\tau(\bm{b}))] = g^{(\tau)}(\bm{b})$. Thus, we will make use of the following ensemble:
\begin{equation}
    \label{eqn:sumcircuits}
    \sum_{\tau \in \mathcal{S}_{n}} \Tr[\bm{O}^{(\tau)}\rho_{\text{QE}}(\tau(\bm{b}))] = \sum_{\tau \in \mathcal{S}_n}g^{(\tau)}(\bm{b}).
\end{equation}
Since each $\frac{\widehat{g}(\bm{s})}{k_{\bm{s}}}\chi_{\bm{s}}(\bm{b})$ appears $k_{\bm{s}}$ times in the sum in Equation \eqref{eqn:sumcircuits} the result follows.
\end{proof}

Similar to the ensemble of phase-embedding-based models, we can utilize a validation data set to determine if the size of the ensemble is sufficient. In addition, we note that a model that makes use of $\bm{U}_{\text{QE}}$ may be less susceptible to overfitting due to higher-order Fourier basis elements not being accessible. We note that Theorem \ref{thm:phase_esemb} implies that for any function with degree $d \leq \frac{n}{3}$, there exists an ensemble of $\binom{n}{\nicefrac{n}{3}}$  phase embedding-based models using $\frac{n}{3}$ qubits that can express the function. Since of course $n!$, i.e. the cardinality of $\mathcal{S}_{n}$, is larger than $\binom{n}{\nicefrac{n}{3}}$, an ensemble of phase embedding models would be more desirable in this case. However, both sufficient conditions still require factorially many models, which can become intractable. We leave as an open question if a smaller ensemble of QRAC-based models is sufficient for expressing interesting functions with $d \leq \frac{n}{3}$ 

We note that a single QRAC-embedding based model still has some beneficial properties.
For example, a single phase-embedding-based model using $m < n$ qubits can only contain Fourier terms that involve $m$ out of the $n$ input variables, i.e. is an $m$-junta. However, a single QRAC-based model can express functions that are dependent on every input variable. 

Lastly, a single linear quantum model using multiple consecutive QRAC embeddings can express a larger class of functions than what was mentioned in Theorem \ref{thm:qrac}. Consider replacing $\bm{U}_{\text{QE}}$ with 
\begin{eqnarray}
\label{eqn:qracrepeat} 
    \prod_{k=1}^{r}\bm{V}_k\bm{U}_{\text{QE}}(\bm{b}),
\end{eqnarray}
where the $\bm{V}_k$ are arbitrary unitary operators that are may or may not be trainable, and there are $r$ data-encoding steps. In Appendix \ref{sec:qrac_improving_through_rep}, we present a concrete example of the unitary operators in Equation \eqref{eqn:qracrepeat} that produces a linear quantum model on a single qubit whose output is expressible as a nontrivial linear combination of all Fourier basis elements for $\mathbb{B}^3$. This alternative operator, in the case where $\bm{V}_{k}$ are not trainable, could be used in place of $\bm{U}_{3,1}$ in Equation \eqref{eqn:qrac_embedding_unitary} in the multiqubit case. However, the degree of freedom that the trainable part of the model, $\bm{O}_{\bm{\theta}}$, has in choosing the coefficients of the $\chi_{\bm{P}}$ is limited when compared to the ensemble approach.

\section{Experiments}
\label{sec:experiments}
We present some experiments, in simulation and on hardware, to demonstrate scenarios in which it is possible to use the phase/QRAC embeddings in a variational linear quantum model to fit low-degree functions on the Boolean cube. All experiments were performed utilizing the Qiskit \cite{qiskit} machine learning framework. The code for executing the experiments in simulation is available online at \href{https://doi.org/10.5281/zenodo.7805753}{https://doi.org/10.5281/zenodo.7805753}. The goal is to show the expressivity of the models, i.e. demonstrating the theory in action, rather than assessing their ability to generalize to unseen data. Thus we provide the models access to all of the data to train on. More explicitly the training set is $\mathcal{T} = \{ (\bm{b}, g(\bm{b}))~|~\forall \bm{b} \in \mathbb{B}^n\}$ for fitting the target function $g: \mathbb{B}^n \xrightarrow[]{} \mathbb{R}$. As discussed in Section \ref{sec:quantumcircuitlearning}, the goal of such a supervised learning task is to minimize Equation \eqref{eqn:regularized_risk}. The loss function utilized for each experiment below is the square error defined as
\begin{equation}
    \ell(\bm{b}, g(\bm{b}), f_{\bm{\theta}}) = \frac{1}{2}(g(\bm{b}) - f_{\bm{\theta}}(\bm{b}))^2,
\end{equation}
where $f_{\bm{\theta}}$ is the model and $g$ is the target function. We did not utilize regularization in any experiment, and thus the regularization term, in Equation \eqref{eqn:regularized_risk}, is zero.  In the QRAC embedding case, for simplicity, we only make use of a single linear quantum model instead of an ensemble. Employing the notation from the previous sections, for all experiments, $\bm{W}(\bm{\theta})$ is an $m$-qubit PQC consisting of single-qubit rotation gates, $\bm{R}_{\mathsf{Y}}$ and $\bm{R}_{\mathsf{Z}}$, and two-qubit controlled-$\mathsf{Z}$ gates using nearest-neighbor connections. Lastly, $\bm{O}_{\bm{\theta}}=\bm{W}^{\dagger}(\bm{\theta})\mathsf{Z}^{\otimes m}\bm{W}(\bm{\theta})$. Here we have chosen $\bm{D}$ from Section \ref{sec:quantumcircuitlearning} to be $\mathsf{Z}^{\otimes m}$. Such a selection of $\bm{D}$ happens to be sufficient for the functions we consider in our experiments. As mentioned in Section \ref{sec:phase} this is not sufficient in general for either embedding. The functions were chosen this way so that the number of circuit runs on hardware could be reduced. The goal is to find a parameter setting for $\bm{\theta}$ such that $\bm{O}_{\bm{\theta}}$ implements an observable $\bm{O}^{(g)}$ satisfying the property:
\begin{align}
    \forall \bm{b} \in \mathbb{B}^n: \Tr[\bm{O}^{(g)}\rho(\bm{b})] = g(\bm{b}).
\end{align}

For each simulated and experimental result we display the functions' values for different Boolean inputs as well as the Fourier coefficients of the learned quantum model in order to show the alignment between the predicted values and the experimental results. For both embeddings, using the final values of the parameters obtained at the end of training, we classically computed the matrix for $\bm{W}(\bm{\theta})$, which corresponds to the trainable part of our model. Subsequently, we computed the matrix for $\bm{W}^{\dagger}(\bm{\theta})\mathsf{Z}^{\otimes m}\bm{W}(\bm{\theta})$, which equals $\bm{O}_{\bm{\theta}}$.
The Fourier coefficients of the linear quantum models were  computed using Equations \eqref{eqn:phase_fourier_coeffs} and \eqref{eqn:qrac_fourier_coeffs} and the matrix for $\bm{O}_{\bm{\theta}}$. Because the number of circuits ($2~\times $ number of parameters $\times $ size of training set $\times$ number of iterations) scales quickly for implementing optimization with the parameter-shift rule, we utilized the COBYLA \cite{powell1994direct} optimizer instead of standard parameter-shift rules and minibatch learning for both simulation and hardware experiments. The figures that follow later clearly show both embeddings were able to fit the target function of $n$ bits, with the QRAC embedding using only one-third of the qubits compared to the $n$-qubit phase embedding.

In Figure \ref{fig:3_bit_results} we show  experiments utilizing both the phase and QRAC embeddings to fit the function
\begin{equation} 
    \label{definitionf3}
    g_3(\bm{b}) = a_1(-1)^{b_1} + a_2(-1)^{b_2} + a_3(-1)^{b_3}.
\end{equation}
This functional form was chosen because, as shown in Section \ref{sec:qrac}, a single variational linear quantum model using QRAC without permuting the input can represent at most a degree $1$ function using a single qubit. The values of the coefficients, $a_i$, were chosen so that setting $\bm{D} = \mathsf{Z}^{\otimes m}$ would be sufficient to express $g_3$.

Three qubits were used in the phase embedding case and one qubit was used in the QRAC embedding case. The circuits that we used are displayed in  Figure \ref{fig:circuits_3bit}. For the hardware experiments, we applied readout-error mitigation and dynamic decoupling \cite{viola1998dynamical} implemented within Qiskit.
\begin{figure*}[!htb]
	\centering
	\includegraphics[scale=1]{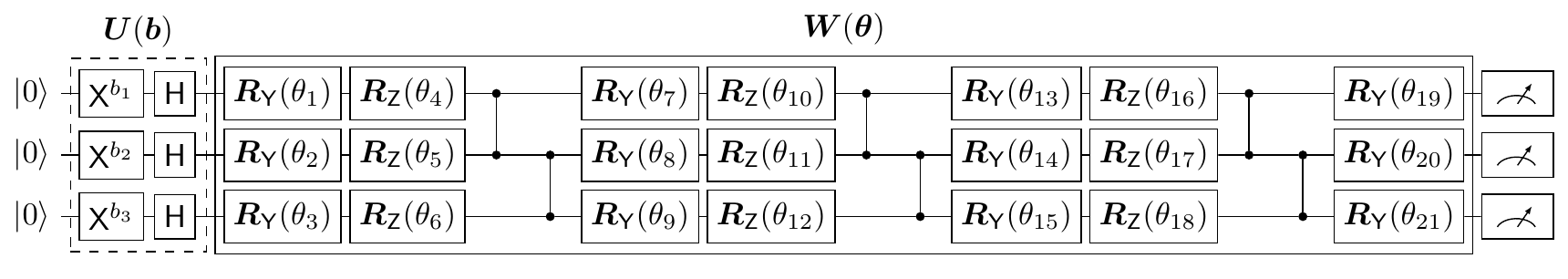}
	\includegraphics[scale=1]{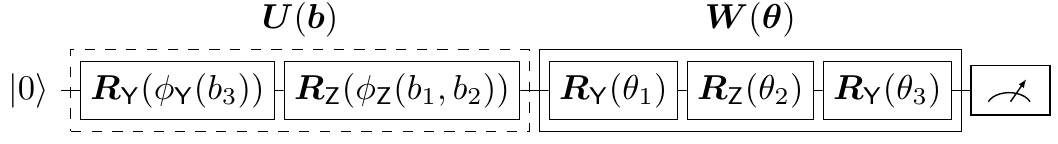}
	\caption{Circuits used in experiments to fit real-valued functions on $\mathbb{B}^3$: on the top is the phase embedding circuit and on the bottom is the QRAC embedding circuit.}
\label{fig:circuits_3bit}
\end{figure*}
Simulation was performed utilizing the statevector simulator. The hardware experiments were performed on the $16$-qubit \textit{ibmq\_guadalupe} device. The phase embedding circuit used qubits $5$, $8$ and $9$ and $300$ iterations of the COBYLA optimizer, and the QRAC embedding circuit used qubit $8$ and $150$ iterations of the COBYLA optimizer. We executed $10,000$ shots for each experiment so that readout-error mitigation could be applied.
\begin{figure*}[!htb]
\centering

\subfloat[Phase Embedding]{
    \includegraphics[width=.48\linewidth]{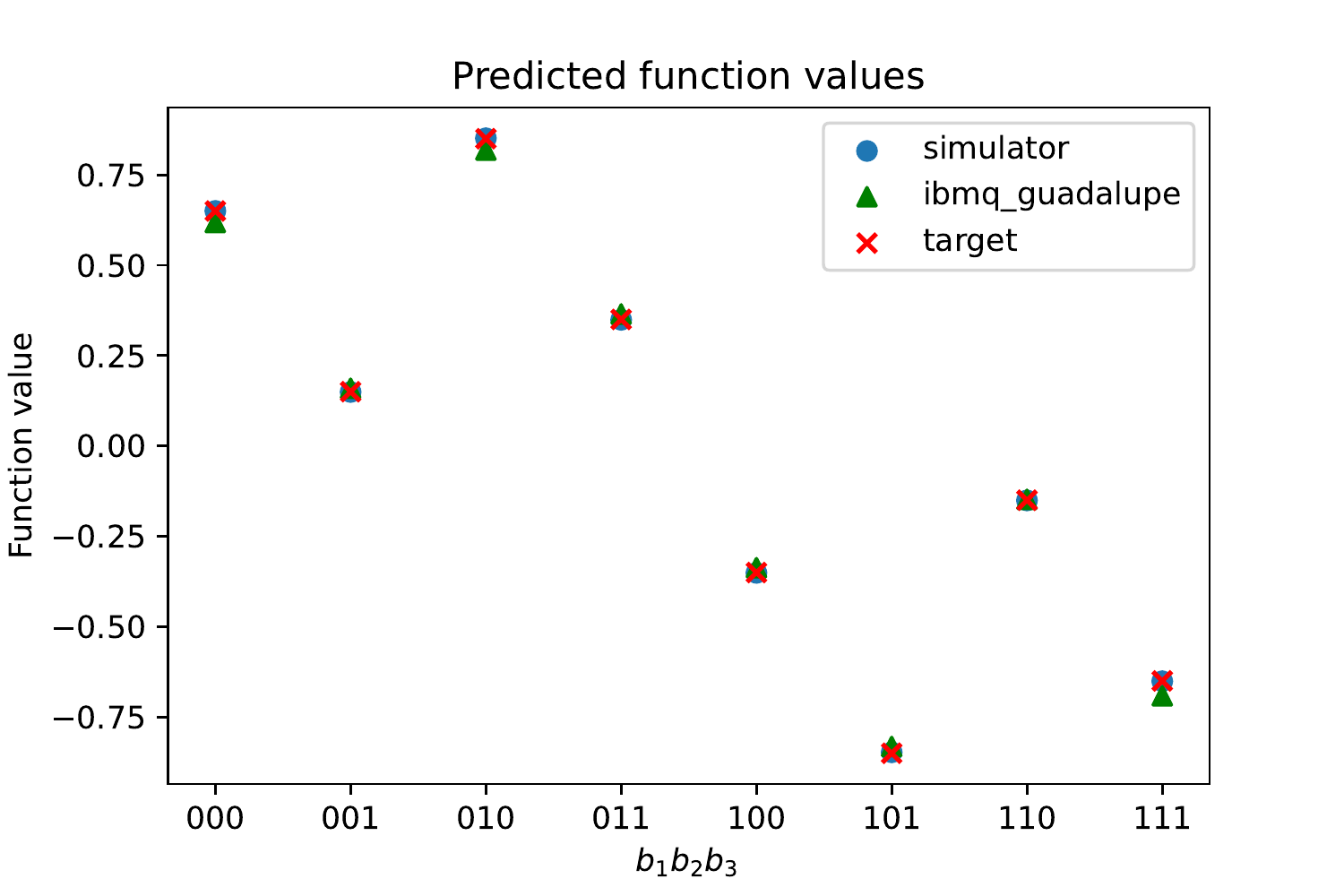}  
    \includegraphics[width=.48\linewidth]{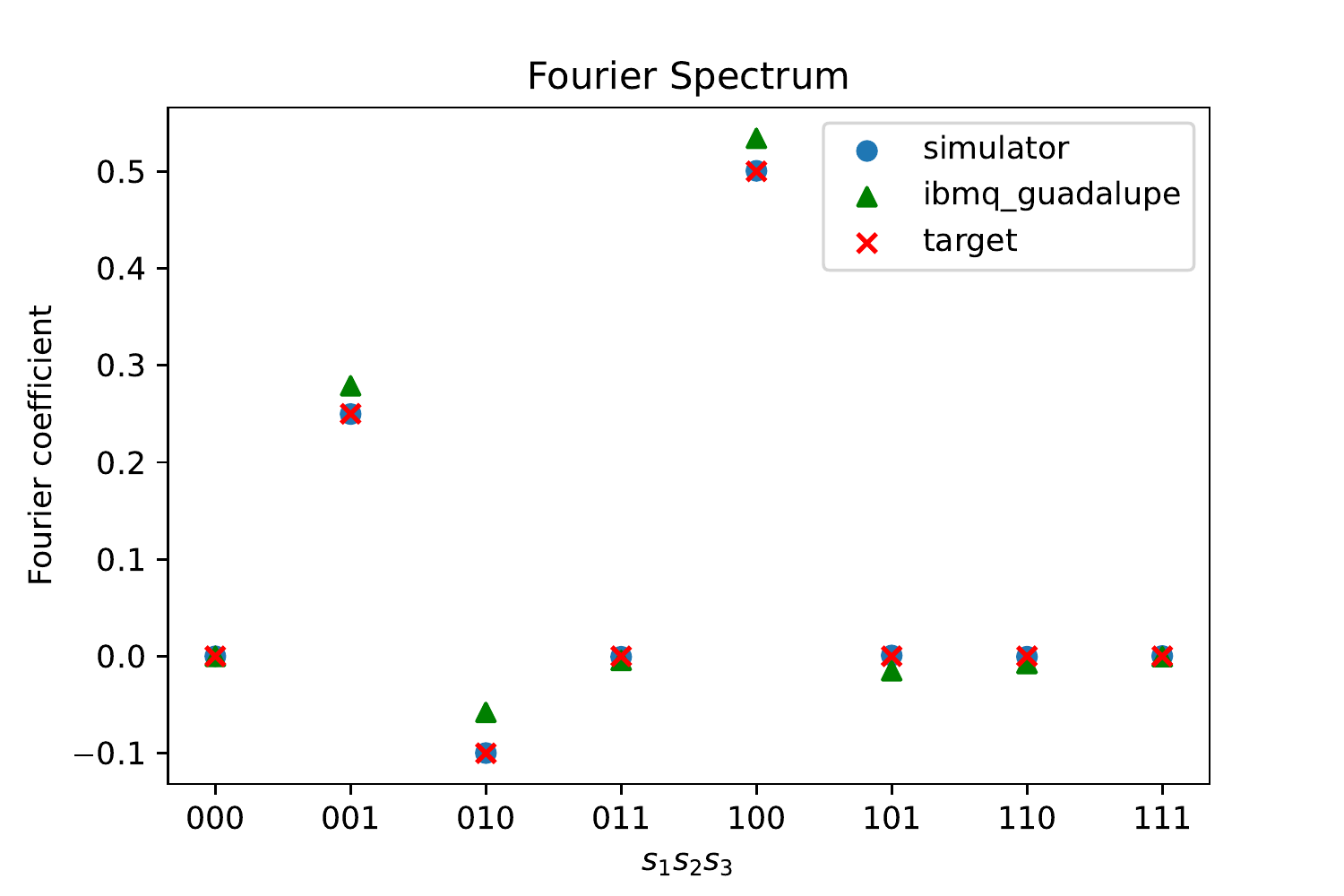}
}\\

\subfloat[QRAC Embedding]{
    \includegraphics[width=.48\linewidth]{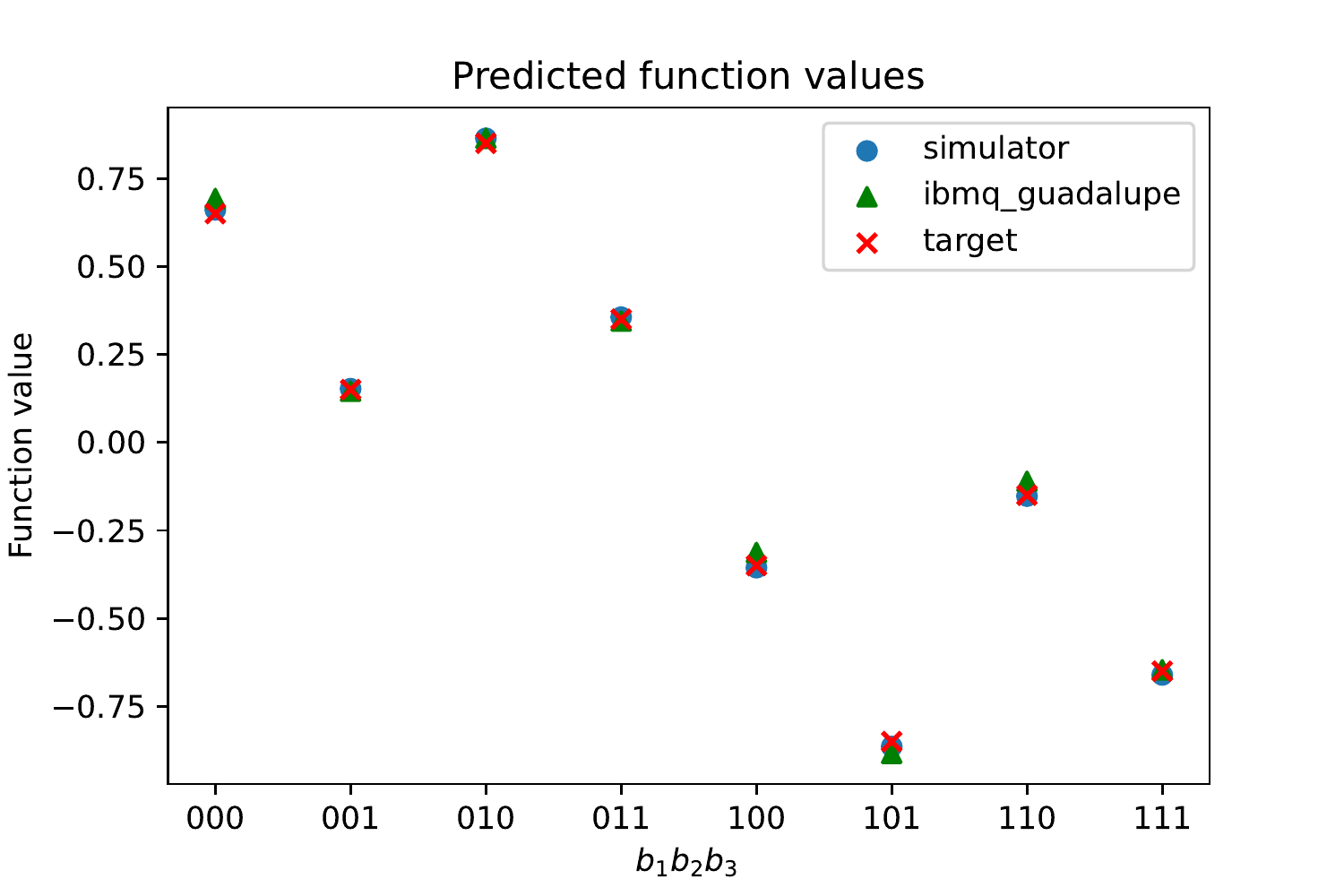}
     \includegraphics[width=.48\linewidth]{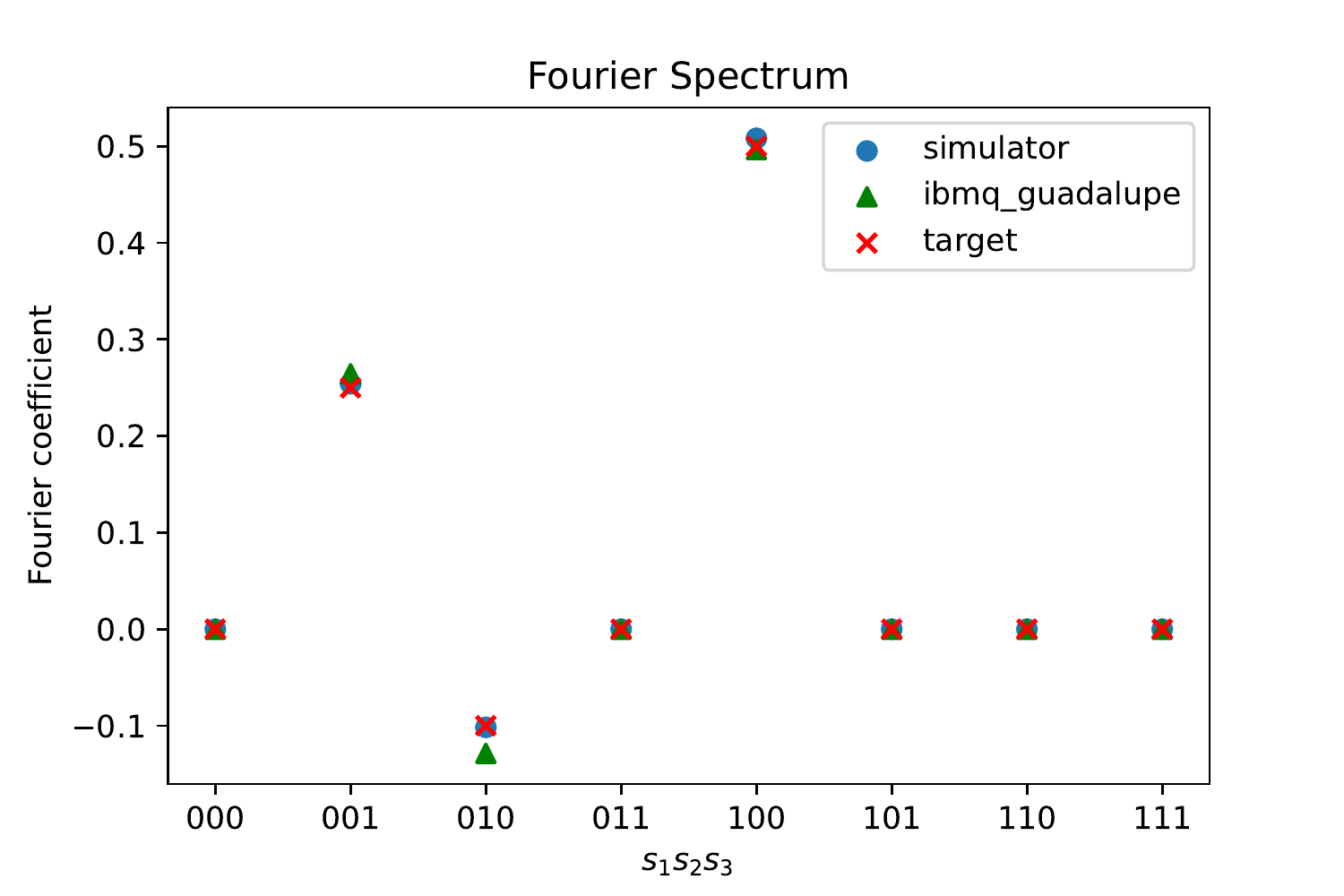}
}

\caption{Simulator and experimental results obtained from  using (a) the phase embedding with three qubits and (b) the QRAC embedding with one qubit to fit the function $g_3$ with $a_1=\frac{1}{2}, a_2=-\frac{1}{10}, a_3 = \frac{1}{4}$. ``Target'' represents the exact outputs and Fourier spectrum of  $g_3$. Both methods successfully fit the target function with high accuracy.}
\label{fig:3_bit_results}
\end{figure*}

Similar to the experiments shown above for the function $g_3(\bm{b})$ with 3-bit inputs, in Figure \ref{fig:6_bit_results}, we present experimental results for learning the following function that depends on 6-bit inputs:
\begin{align} \label{definitionf6} 
    g_6(\bm{b}) = d_1(-1)^{b_1 + b_4} + d_2(-1)^{b_1 + b_5}\nonumber\\+ d_3(-1)^{b_2 + b_4} + d_4(-1)^{b_2 + b_5}.
\end{align}
The functional form of $g_6$ was chosen for similar reason that $g_3$ was chosen in the previous experiment. A variational linear quantum model using QRAC on two qubits  without permuting the input can express at most degree $2$ functions. The coefficients were again chosen so that $\mathsf{Z}^{\otimes m}$ would be sufficient as an observable.
The circuits used are presented in Figure \ref{fig:circuits_6bit}. Here six qubits were used for the phase embedding case while two qubits were used for the QRAC embedding case. Simulation was performed utilizing the statevector simulator. The hardware experiment for the QRAC embedding case was performed on the $7$-qubit \textit{ibmq\_casablanca} device. The circuit used qubits $1$ and $2$ and $200$ iterations of the COBYLA optimizer. Again, $10,000$ shots were executed for each experiment so that readout-error mitigation could be applied.
In this experiment we again observed close agreements between the predictions based on the theory and the experimental results.

\begin{figure*}[!htb]
	\centering
	\includegraphics[scale=0.79]{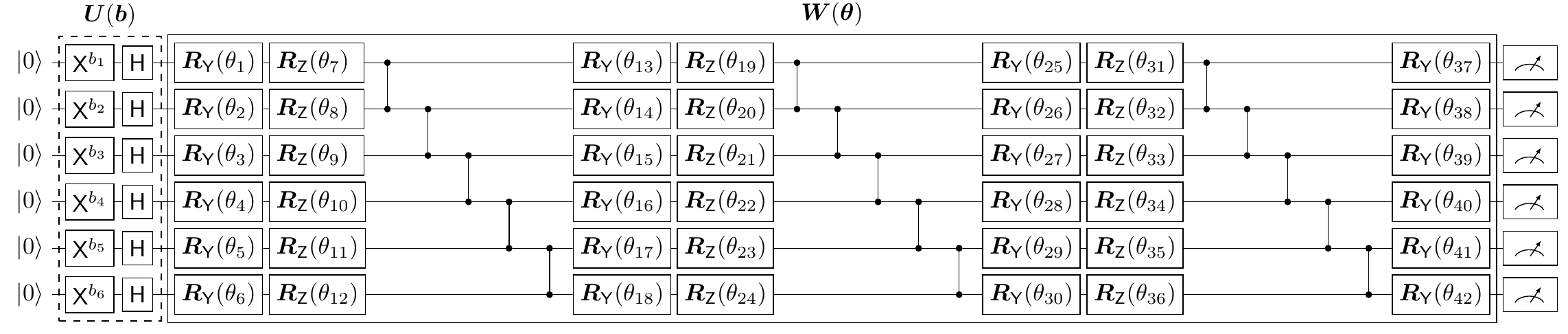}
	\includegraphics[scale=0.79]{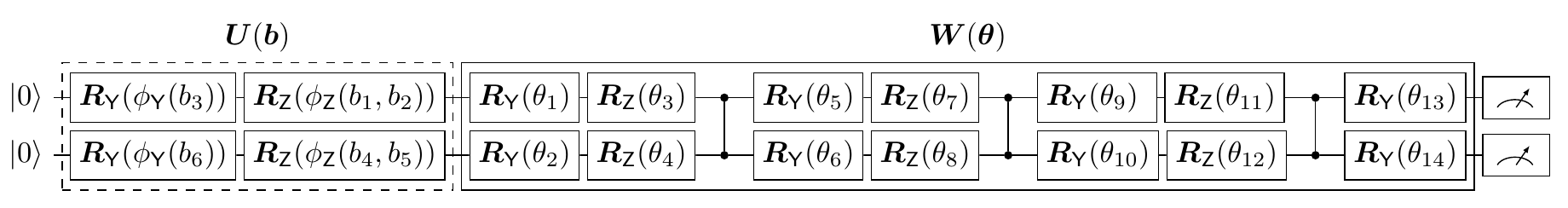}
	\caption{Circuits used in experiments to fit real-valued functions on $\mathbb{B}^6$: on the top is the phase embedding circuit and on the bottom is the QRAC embedding circuit.}
    \label{fig:circuits_6bit}
\end{figure*}

\begin{figure*}[!htb]
\centering
 \subfloat[Phase Embedding]{
	\includegraphics[scale=0.5]{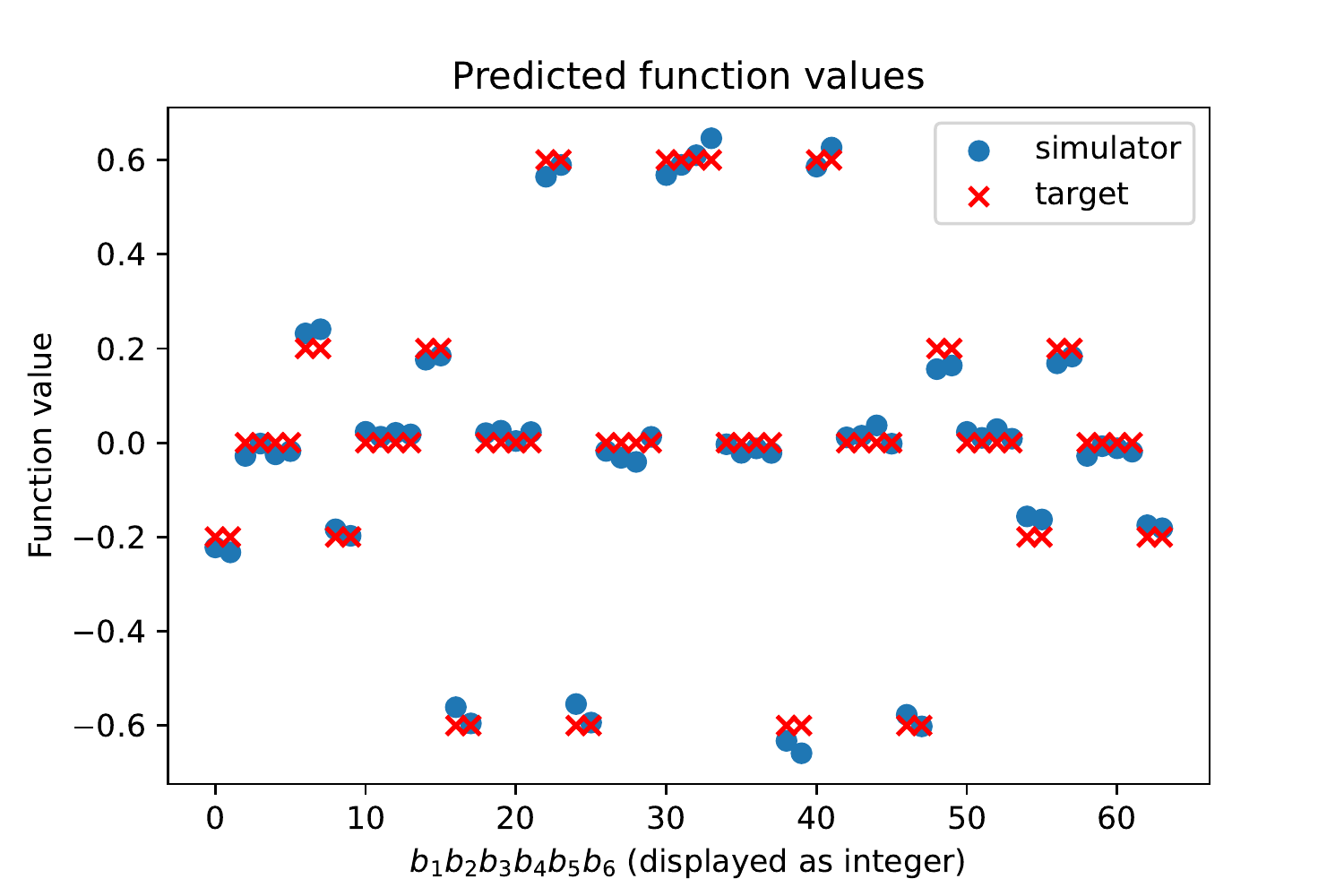}
	\includegraphics[scale=0.5]{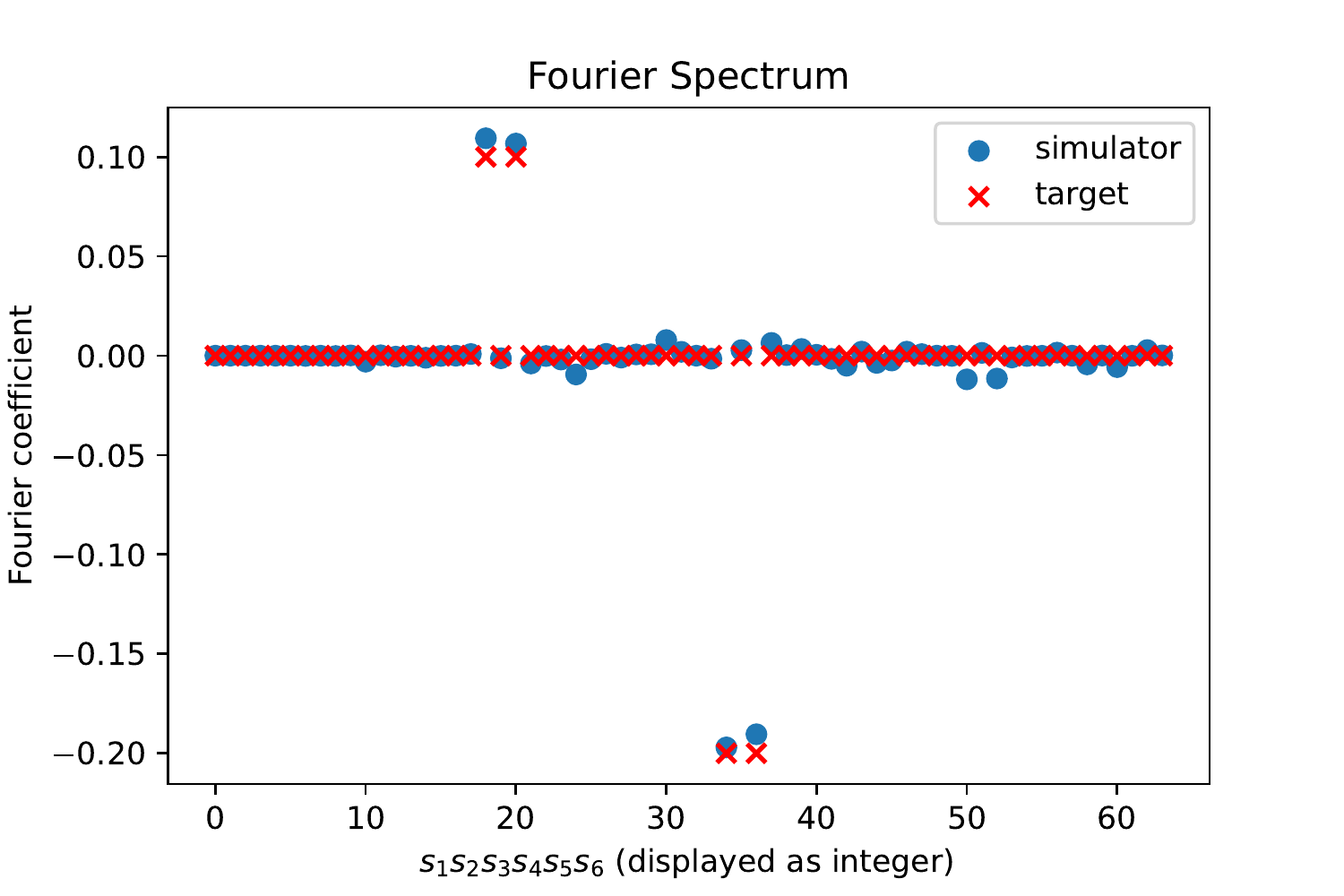}
 }\\
 
\subfloat[QRAC Embedding]{
	\includegraphics[scale=0.5]{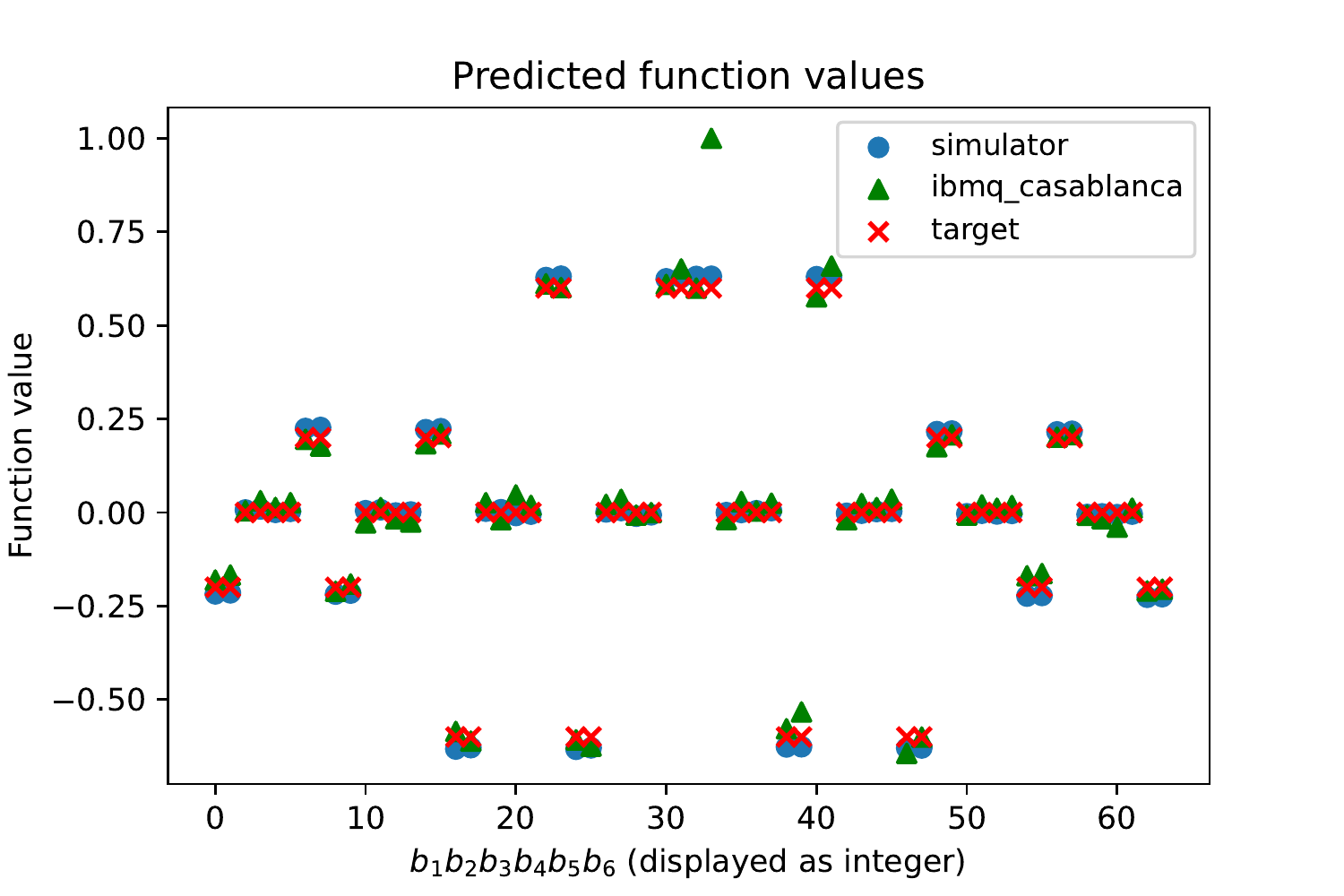}
	\includegraphics[scale=0.5]{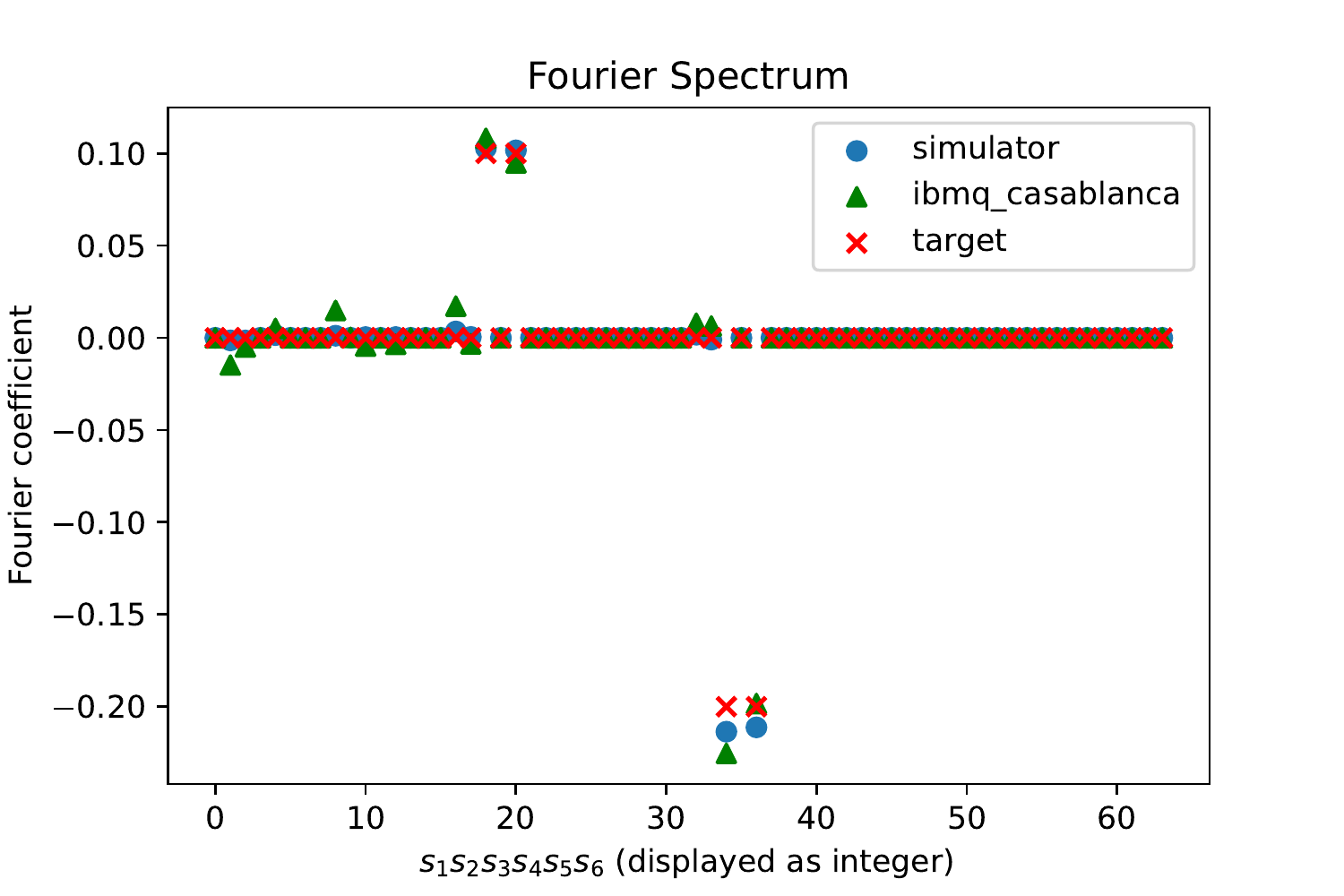}
 }
\caption{Simulator and experimental results obtained from using (a) the phase embedding with six qubits and (b) the QRAC embedding with two qubits to fit the function $g_6$ with $d_1= d_2=-\frac{1}{5}, d_3 =d_4 = \frac{1}{10}$. ``Target'' represents the exact outputs and Fourier spectrum of  $g_6$. Both methods successfully fit the target function with high accuracy.}
\label{fig:6_bit_results}
\end{figure*}

\section{Conclusion}
\label{sec:conclusion}
We summarize the results obtained here and give a few remarks on the implications of our findings. First, we have used Fourier analysis to provide sufficient conditions for a function on the Boolean cube to be expressible via variational linear quantum models or ensembles of variational linear quantum models utilizing the phase and QRAC embeddings. We showed that for any function on the Boolean cube there exists a variational linear quantum model based on the phase embedding that can represent it (Theorem \ref{thm:phase}) and an ensemble of such models that can represent any degree $d$ function with $d$ qubits (Theorem \ref{thm:phase_esemb}). These result narrows down sufficiency conditions for the representability of functions on the Boolean cube. Previously known results were proven for functions in $L_2([0, 2\pi]^n)$, where representability sufficiency was achieved outside of the linear model framework. This was done by showing that repeating the phase embedding $r$-times sequentially (data re-uploading) or in parallel approximates the $r$-th cubic partial sum of  a function's Fourier series (Equation \eqref{eqn:partial_fourier}).

We then showed, via Theorem \ref{thm:qrac} and Theorem \ref{thm:permuted_spectra}, that a single linear quantum model using the QRAC embedding can express low-degree functions on $\mathbb{B}^n$, if the function satisfies the property that the Fourier coefficent of $\chi_{\bm{s}}$ being nonzero implies that  $\bm{s} \in \mathcal{K}^{\text{QE}}_{\frac{n}{3}}$, or in $\tau(\mathcal{K}^{\text{QE}}_{\frac{n}{3}})$ if we permute the input by $\tau \in \mathcal{S}_n$. Lastly, we demonstrated that ensembles of  linear quantum models that use quantum random access codes can represent functions on the Boolean cube with degree $d \leq \lceil\frac{n}{3}\rceil$ (Theorem \ref{thm:qrac_ensemb}). The variational linear quantum models presented for learning functions on the Boolean cube can be easily applied to problems involving other discrete domains by converting integer representations to binary. Machine learning problems involving discrete-valued inputs appear frequently in industrial settings. For example, categorical features are known to be essential for machine learning tasks in financial \cite{DUAN20194716} and healthcare applications \cite{Hancock2020}.  

The results presented can be expanded in different directions. Future research can benchmark  model ensembles that use the phase or QRAC embeddings. %
It would be interesting to further study the impact of classical preprocessing on VQML models, which we showed to be beneficial for both embeddings. Potentially, similar expressivity theorems, like those in Section \ref{sec:embeddings}, can be demonstrated for linear quantum models that operate on discrete domains beyond the Boolean cube. For example, quantum computation is already known to provide significant computational speedups for problems involving finite Abelian groups \cite{2010childsalgebra}.

Furthermore, subsequent work could also compare the expressivities of classical neural networks to VQML models. One could obtain an upper bound on the required size of the neural network using the fact that an arbitrary real-valued function on the Boolean cube is a linear combination of parities.  It is folklore that a single hidden layer of size two suffices to express $\mathsf{XOR}$ on two input bits. Thus if one uses a divide-and-conquer approach, then a $d$-bit parity can be expressed with $2\log(d)$ hidden layers of width at most $d$. These form a binary tree with each layer repeated, and thus uses $\mathcal{O}(d)$ neurons. To express an arbitrary function,  the widest hidden layer would be at most the input width times the number of nonzero Fourier coefficients. 

This appears to be comparable to VQML case. For a function with an exponentially large set of nonzero Fourier coefficients, the  neural network may require exponentially many neurons. In this case, VQML may require a diagonal observable that decomposes into exponentially-many elements of $\text{span}_{\mathbb{R}}\{\mathbb{I}, \mathsf{Z}\}^{\otimes n}$ and require an exponentially deep PQC. However, the Fourier space representation of the function may not be the most computationally efficient form, and thus the classical neural network could use fewer resources.  Nevertheless, we note that uniquely in quantum, one can have a \textit{trivial} learner from Fourier sampling that may give quantum advantage~\cite{BshoutyJackson1998} providing access to uniform quantum examples (see also Appendix A of~\cite{Arunachalam2022}).
We leave a detailed comparison of these models as the topic of future work.

We performed proof-of-principle numerical experiments and executed the algorithms on IBM quantum processors. These experiments demonstrated that it is possible for a variational linear quantum model using  the embeddings presented to learn sufficient parameters to express low-degree functions. 
In future developments, one could study the ability for such models to generalize to unseen data, i.e. truly learn, and quantify the required number of training samples needed to learn functions, such as low-degree $k$-juntas. For simplicity, we setup the problem scenarios so that $\bm{D}=\mathsf{Z}^{\otimes m}$ was sufficient for all learning tasks. However, it would be interesting to experiment with more complicated problems where such a simple observable does not suffice. Potentially, there exist interesting classes of functions that can be expressed with a fixed $\bm{D}$ that is a linear combination of only polynomially (in the number of qubits) many Pauli terms. Lastly, there might be cases where we can exploit the structure of the problem to design efficient PQCs, particularly for near-term quantum hardware, for learning functions on the Boolean cube.
\vspace{10mm}

\appendices
\section{Repeated Phase Embeddings}\label{sec:repeated_embeddings}
\subsection{Phase Embedding}\label{sec:phase_improving_through_rep}

In this appendix, we show that incorporating serial repetitions into the phase embedding increases the expressivity of a model that uses it. Consider the embedding
\begin{equation}
     \bm{U}_{\operatorname{PE}}^{(m, r)}(\bm{b}) := \prod_{j=1}^{r}\mathsf{Z}^{(\nu_{\bm{w}_j}(\bm{b}))}\bm{V}_{\bm{w}_j}.
\end{equation}
where $\bm{V}_{\bm{w}_1} = \mathsf{H}^{\otimes m}$. Then, it follows that

\begin{align}
    &\bra{0_m}(\bm{U}_{\operatorname{PE}}^{(m, r)})^{\dagger}(\bm{b})\bm{O}\bm{U}_{\operatorname{PE}}^{ (m,r)}(\bm{b})\ket{0_m}\nonumber\\
    &=\bra{0_m}(\prod_{j=1}^{r}\mathsf{Z}^{(\nu_{\bm{w}_j}(\bm{b}))}\bm{V}_{\bm{w}_j})^{\dagger}\bm{O}\prod_{j=1}^{r}\mathsf{Z}^{(\nu_{\bm{w}_j}(\bm{b}))}\bm{V}_{\bm{w}_j}\ket{0_m} \nonumber\\
    &= \bra{+_n}\mathsf{Z}^{(\nu_{\bm{w}_1}(\bm{b}))}\bm{V}^{\dagger}_{\bm{w}_2}\mathsf{Z}^{(\nu_{\bm{w}_2}(\bm{b}))}\cdots \bm{V}^{\dagger}_{\bm{w}_r}\mathsf{Z}^{(\nu_{\bm{w}_r}(\bm{b}))}\bm{O}\mathsf{Z}^{(\nu_{\bm{w}_r}(\bm{b}))}\nonumber\\&\bm{V}_{\bm{w}_r}\cdots\mathsf{Z}^{(\nu_{\bm{w}_2}(\bm{b}))}\bm{V}_{\bm{w}_2}\mathsf{Z}^{(\nu_{\bm{w}_1}(\bm{b}))}\ket{+_n} \nonumber\\
    &= \frac{1}{2^m}\sum_{(\bm{y})_t, (\bm{k})_t \in (\mathbb{B}^m)^{\times r}}\Bigg[(-1)^{\sum_{t}\nu_{\bm{w}_t}(\bm{b}) \cdot (\bm{y})_t}(\bm{V}^{\dagger}_{(\bm{w})_{2:r}})_{(\bm{y})_{1:r}}\nonumber\\&\bm{O}_{\bm{y}_r,\bm{k}_r}(\bm{V}_{(\bm{w})_{2:r}})_{(\bm{k})_{1:r}}(-1)^{\sum_{t}\nu_{\bm{w}_t}(\bm{b}) \cdot (\bm{k})_t}\Bigg] \nonumber\\
    &=\sum_{(\bm{y})_t, (\bm{k})_t \in (\mathbb{B}^m)^{\times r}}\bm{W}_{(\bm{y})_t, (\bm{k})_t}(-1)^{\sum_{t}\nu_{\bm{w}_t}(\bm{b})\cdot ((\bm{y})_{t} \oplus (\bm{k})_{t})},
\end{align}

where 
\begin{equation}
(\bm{V}_{(\bm{w})_{2:r}})_{(\bm{k})_{1:r}} := (V_{\bm{w}_r})_{\bm{k}_r, \bm{k}_{r-1}}\cdots(V_{\bm{w}_2})_{\bm{k}_2, \bm{k}_{1}}
\end{equation}

and
\begin{equation}
\bm{W}_{(\bm{y})_t, (\bm{k})_t}:=\frac{1}{2^m}(\bm{V}^{\dagger}_{(\bm{w})_{2:r}})_{(\bm{y})_{1:r}}\bm{O}_{\bm{y}_r,\bm{k}_r}(\bm{V}_{(\bm{w})_{2:r}})_{(\bm{k})_{1:r}}.
\end{equation}

Suppose that $m = \frac{n}{r}$ and that each $\nu_{\bm{w}_j}$ partitions the $n$ inputs bits into $r$-tuples of size $m$. Then it follows that the above reduces to
\begin{equation}
    \sum_{\bm{s} \in \mathbb{B}^n}\widetilde{\bm{W}}_{\bm{s}}(-1)^{\bm{b} \cdot \bm{s}} = \sum_{\bm{s} \in \mathbb{B}^n}\widetilde{\bm{W}}_{\bm{s}}\chi_{\bm{s}}(\bm{b}),
\end{equation}
where 
\begin{align}
    \widetilde{\bm{W}}_{\bm{s}} := \sum_{\substack{(\bm{y})_t, (\bm{k})_t \in (\mathbb{B}^m)^{\times r} \\ (\bm{y})_{t}\oplus(\bm{k})_{t} = \bm{s}}}\bm{W}_{(\bm{y})_t, (\bm{k})_t}
\end{align} 
Note that the set
\begin{align}
    \{&(\bm{y})_{t}\oplus(\bm{k})_{t} = (y_{11} \oplus k_{11}, y_{12} \oplus k_{12}, \dots, y_{1m} \oplus k_{1m}, \nonumber\\
&\quad \dots, y_{j1} \oplus k_{j1}, \dots, y_{rm} \oplus k_{rm})~|\nonumber\\&\quad
    (\bm{y})_t, (\bm{k})_t \in (\mathbb{B}^m)^{\times r}\}
\end{align}
contains all elements of $\mathbb{B}^n$. Thus, it is possible for the Fourier spectrum of this model to have support on any of the Fourier basis elements, which implies an increase in expressivity.

\subsection{QRAC Embedding}
\label{sec:qrac_improving_through_rep}

In this appendix, we present an example that shows that using multiple consecutive QRAC embeddings does enrich the class of functions that a single linear quantum model using this embedding can represent. Let $\bm{R}_{\bm{n}}(\theta) = e^{-i\frac{\theta}{2}(n_1\mathsf{X} +n_2\mathsf{Y} + n_3\mathsf{Z})}$, where $\bm{n} \in \mathbb{R}^3$ and  $\lVert \bm{n}\rVert_2 = 1$. We will consider replacing $\bm{U}_{3,1}$ with 
\begin{equation}
\tilde{\bm{U}}_{3,1} = \bm{U}_{3,1}\bm{R}_{\bm{n}}(\pi)\bm{U}_{3,1},
\end{equation}
where $\bm{n} = \frac{1}{\sqrt{3}}(1, 1, 1)$.
Then,
\begin{align}
    f_{\bm{\theta}}(\bm{b}) &= \Tr[O_{\bm{\theta}}\tilde{\bm{U}}_{3,1}(\bm{b})\ketbra{0}\tilde{\bm{U}}_{3,1}^{\dagger}(\bm{b})]  \nonumber\\
    &= a_1\Tr[\bm{O}_{\bm{\theta}}] + (a_2 + a_3(-1)^{b_1 + b_2 + b_3} + a_4(-1)^{b_1 + b_2} \nonumber\\&+ a_5(-1)^{b_1 + b_3} + a_6(-1)^{b_1} + a_7(-1)^{b_2 + b_3})\Tr[\bm{O}_{\bm{\theta}}\mathsf{X}] \nonumber\\
    &+ (a_8 + a_9(-1)^{b_3} + a_{10}(-1)^{b_1 + b_2} + a_{11}(-1)^{b_1 + b_2 + b_3} \nonumber\\&+ a_{12}(-1)^{b_2 + b_3} + a_{13}(-1)^{b_2} + a_{14}(-1)^{b_1 + b_3})\Tr[\bm{O}_{\bm{\theta}}\mathsf{Y}] \nonumber \\ 
    &+ (a_{15} + a_{16}(-1)^{b_1 + b_3} + a_{17}(-1)^{b_2 + b_3} \nonumber \\& + a_{18}(-1)^{b_1} + a_{19}(-1)^{b_2})\Tr[\bm{O}_{\bm{\theta}}\mathsf{Z}],
\end{align}
where the $a_i$ are fixed, but we have abstracted them out because our focus is on the number of Fourier basis terms, $(-1)^{\bm{s}\cdot\bm{b}}$. This model is a linear combination of all Fourier basis elements for functions on $\mathbb{B}^3$. 
The trainable component of the model determines the values for $\Tr[\bm{O}_{\bm{\theta}}\bm{P}]$, where $\bm{P}$ is a Pauli operator. For a given Pauli operator this value is shared by more than one Fourier basis element. 

For obtaining the expansion above, it is helpful to express the $\bm{R}_{\mathsf{Z}}$ and $\bm{R}_{\mathsf{Y}}$ rotations involved in $\bm{U}_{3,1}$ in terms of $b_1, b_2, b_3$ as follows:
\begin{align}
     &\bm{R}_{\mathsf{Z}}(\phi_{\mathsf{Z}}(b_1, b_2)) = \gamma_{+}(b_1)(c^{(\mathsf{Z})}_{+}\mathbb{I} - c^{(\mathsf{Z})}_{-}i(-1)^{b_2}\mathsf{Z}) \nonumber\\&\quad+ \gamma_{-}(b_1)(c^{(\mathsf{Z})}_{-}\mathbb{I} - c^{(\mathsf{Z})}_{+}i(-1)^{b_2}\mathsf{Z}) \\ 
    &\bm{R}_{\mathsf{Y}}(\phi_{\mathsf{Y}}(b_3)) = \gamma_{+}(b_3)(c^{(\mathsf{Y})}_{+}\mathbb{I} - c^{(\mathsf{Y})}_{-}i\mathsf{Y}) \nonumber\\&\quad+ \gamma_{-}(b_3)(c^{(\mathsf{Y})}_{-}\mathbb{I} - c^{(\mathsf{Y})}_{+}i\mathsf{Y}),
\end{align}
where
\begin{align}
    c^{(\mathsf{Z})}_{\pm} = \frac{\sqrt{2 \pm \sqrt{2}}}{2} \quad \textnormal{and} \quad \quad c^{(\mathsf{Y})}_{\pm} = \sqrt{\frac{1}{2} \pm \frac{1}{2\sqrt{3}}}
\end{align}
as well as 
\begin{align}
    \gamma_{\pm}(b) = \frac{1 \pm (-1)^{b}}{2}.
\end{align}
The functions $\phi_{\mathsf{Z}}$ and $\phi_{\mathsf{Y}}$ are defined in Section \ref{sec:qrac} with
\begin{align}
    \alpha_1 = \frac{\pi}{4} \quad \textnormal{and} \quad \alpha_2 = 2\cos^{-1}\bigg(\sqrt{\frac{1}{2} + \frac{1}{2\sqrt{3}}}\bigg).
\end{align} 
This formulation introduces dependence on terms of the form $(-1)^{\bm{s}\cdot\bm{b}}$ when expanding the expectation.

\section{Using variational SWAP networks in the phase embedding}
\label{sec:phase_swaps}
When using the phase embedding, after loading the input bits onto a register with  $\mathsf{X}$ gates, we can apply a layer of variational $\mathsf{SWAP}$ gates, i.e. $e^{-i\frac{\beta}{2}\mathsf{SWAP}}$, with learnable parameters $\beta$. The layer consists of one variational $\mathsf{SWAP}$ between every pair of qubits, i.e. $\binom{n}{2} = \frac{n(n-1)}{2}$ gates.  This allows for testing multiple combinations of the $k$ out of $n$ input bits in superposition. Specifically setting all $\beta$'s to $\frac{\pi}{2}$ produces a uniform superposition containing all possible subsets of $k$ bits that can be swapped into the first $k$ bits. One motivation behind adding the $\mathsf{SWAP}$ network is due to the following lemma.

\begin{lemma}
Consider the input to the phase embedding, where the bits are loaded onto a computational basis state. Suppose the parameterized observable $\bm{O}_{\bm{\theta}}$ and the layer of Hadamards acts only on the first $k$ of the $n$ qubits and thus the output of the model only depends on $b_1, \dots, b_k$. For any fixed layout of all-to-all variational $\mathsf{SWAP}$ network inserted after loading the input bits, and any subset $\mathcal{B}$ of $k$ of the input bits, there exists a setting of the variational parameters such that the model depends only on $\mathcal{B}$.
\end{lemma}
\begin{proof}
We can find a bijective mapping between the $k$ relevant bits and the first $k$ qubits, potentially acting as the identity on some qubits. This map can be expressed as a product of transpositions of the $n$ input elements that do not act on the same qubit. Thus, any all-to-all variational $\mathsf{SWAP}$ network can implement this map by enabling/disabling the relevant $\mathsf{SWAP}$s.
\end{proof}

As a proxy for variational $\mathsf{SWAP}$s, one could use the particle-preserving $\mathsf{XY}$ gate \cite{Hadfield_2019}.
The benefit of a variational $\mathsf{SWAP}$ network in practice would require further experimentation.  For QRAC, it appears that we would need to encode the $n$ input bits into an additional quantum register destroying the constant-factor reduction in qubits. Also, the bits-to-angle mapping for QRAC would need to be implemented coherently and the rotation gates controlled on additional ancillas.

\section{Generalization bounds}

While our focus is on expressivity, we can almost trivially apply one of the generalization bounds obtained by \cite{Caro2021encodingdependent} to obtain one for the phase embedding and QRAC embedding. The following is the definition of a variational linear quantum model, $f_{\bm{\theta}}$ that was presented in Section \ref{sec:quantumcircuitlearning}:

\begin{equation}
    \label{eqn:variational_linear_quantum_model_apdx}
    f_{\bm{\theta}}(\bm{b}) := \Tr[\bm{O}_{\bm{\theta}}\rho(\bm{b})],
\end{equation}
where 
\begin{equation}
\label{eqn:param_obs_apdx}
    \bm{O}_{\bm{\theta}} := \bm{W}^{\dagger}(\bm{\theta})\bm{D}\bm{W}(\bm{\theta}).
\end{equation}
The operator $\bm{D}$ is an observable that is diagonal in the computational basis, and $\bm{W}(\bm{\theta})$ is a parameterized-unitary operator. The unitary used to prepare the feature state $\rho(\bm{b})$ can be the phase or QRAC embedding.

\begin{theorem}
\label{cor:phase_embed_gen}
Let $n, m \in \mathbb{N}$, and $\ell:\mathbb{R} \times \mathbb{R} \xrightarrow[]{} [0, c]$ be a loss function that is $\beta$-Lipschitz in the second coordinate. In addition, consider an arbitrary $\delta \in (0, 1)$ and arbitrary probability measure $\mu$ on $\mathbb{B}^{n} \times \mathbb{R}$. Furthermore, suppose $f_{\bm{\theta}}$ is variational linear quantum model using either the phase or QRAC embedding, then, with probability $\geq 1 - \delta$ over the choice of an i.i.d training set $\mathcal{T}$, the model $f_{\bm{\theta}}$ satisfies:

\begin{equation}
    R(f_{\bm{\theta}}) \leq \hat{R}_{\mathcal{T}}(f_{\bm{\theta}}) + \widetilde{\mathcal{O}}\left(\frac{\beta\lVert \bm{D}\rVert_{2} + c\sqrt{\log(1/\delta)}}{\sqrt{|\mathcal{T}|}}\right),
\end{equation}
where $R$ and $\hat{R}_{\mathcal{T}}$ are the generalization error and training error respectively. Additionally, $\tilde{O}$ suppresses poly-logarithmic factors.
\end{theorem}
\begin{proof}
By Equation \eqref{eqn:continuous_phase_embed} the phase embedding can be viewed as a Pauli encoding with restricted domain. Thus Corollary 14 result (a) from \cite{Caro2021encodingdependent} applies, and in the phase embedding case $n$-encoding gates are used for an $n$-dimensional input. Since quantum models based on the QRAC embedding can be viewed as PQCs with Pauli encodings on  a $\frac{2n}{3}$-dimensional domain defined by $\phi_{\mathsf{Z}}$ and $\phi_{\mathsf{Y}}$ using $\frac{2n}{3}$-encoding gates, the same bound applies to QRAC-based models.  Since $\bm{W}$ is unitary,  $\lVert \bm{O}_{\bm{\theta}} \rVert_{2} = \lVert \bm{D} \rVert_{2}$, and so the bound applies to $\bm{O}_{\bm{\theta}}$ too.
\end{proof}
\section{Experimental Device Parameters}
Here we report the experimental device parameters for each of the hardware experiments presented in Section \ref{sec:experiments}. The experiments to fit functions on $\mathbb{B}^3$ were carried out on the 
\textit{ibmq\_guadalupe} device, where qubits 5, 8 and 9 were used for the phase embedding experiment and qubit 8 was used for the QRAC embedding experiment. The experiment to fit a function on $\mathbb{B}^6$ was carried out on the \textit{ibmq\_casablanca} device using qubits 1 and 2.
\begin{center}
\begin{tabular}{ c|c|c|c } 
 \hline
 Parameter & 3-bit phase & 3-bit QRAC & 6-bit QRAC \\ \hline 
 T1 ($\mu$s) & 95 & 130 & 104 \\
 T2 ($\mu$s) & 87 & 98 & 99 \\
 single-qubit error & 0.000283 & 0.000275 & 0.000501 \\ 
 two-qubit error & 0.00673 & 0.00727 & 0.0101 \\
 readout error & 0.017 & 0.0207 & 0.0193\\
 \hline
\end{tabular}
\end{center}

\section*{Acknowledgements}
D.H. would like to thank Yue Sun, Arthur Rattew, Shouvanik Chakrabarti, and Ruslan Shaydulin for insightful discussions. In addition, we would like to thank Pierre Minssen and Shaohan Hu for their feedback on this manuscript.
\section*{Disclaimer}
This paper was prepared with synthetic data and for informational purposes with contributions from the Global Technology Applied Research center of JPMorgan Chase \& Co. This paper is not a product of the Research Department of JPMorgan Chase \& Co. or its affiliates. Neither JPMorgan Chase \& Co. nor any of its affiliates makes any explicit or implied representation or warranty and none of them accept any liability in connection with this paper, including, without limitation, with respect to the completeness, accuracy, or reliability of the information contained herein and the potential legal, compliance, tax, or accounting effects thereof. This document is not intended as investment research or investment advice, or as a recommendation, offer, or solicitation for the purchase or sale of any security, financial instrument, financial product or service, or to be used in any way for evaluating the merits of participating in any transaction.

\bibliographystyle{unsrt}
\bibliography{bibliography}

\EOD

\end{document}